\documentclass[12pt]{article}

\usepackage[T1]{fontenc}
\usepackage[active]{srcltx}
\usepackage{lmodern}
\usepackage{amsmath,amstext,amssymb,graphicx,graphics,color,mathtools}
\usepackage{tikz}
\usepackage{mleftright}
\usepackage{pgfplots}
\usepackage{theorem}
\usepackage{cite}               
\usepackage{hyperref}           
\usepackage{bbm}                
\usepackage{fancybox}           
\usepackage[section]{placeins}  
\usepackage{subcaption}
\usepackage{enumerate}
\theorembodyfont{\normalfont\slshape} 
\newtheorem{definition}{Definition}
\newtheorem{theorem}{Theorem}
\newtheorem{proposition}{Proposition}[section]

\newtheorem{lemma}[proposition]{Lemma}
\theoremstyle{break} 

\newenvironment{proof}%
{{\par\noindent \bf Proof. \nobreak}}%
{\nobreak \removelastskip \nobreak \hfill $\Box$ \medbreak}
{{\par\noindent \bf Proof \nobreak}}%
{\nobreak \removelastskip \nobreak \hfill $\Box$ \medbreak}
{{\par\noindent \bf Proof lemma. \nobreak}}%
{\nobreak \removelastskip \nobreak \bf End proof lemma. \medbreak}
\newenvironment{remark}{\par \medskip \noindent {\bf Remark. }\nobreak}{\par \medskip}

\usepackage{fancyhdr}
 \fancypagestyle{plain}{
  \fancyhead{}
  
}

\def\paragraph#1{{\bf #1\ }}
\newcommand{\RN}[1]{%
  \textup{\uppercase\expandafter{\romannumeral#1}}%
}
\newcommand{\expo}{\mathrm{e}}

\newcommand{\dd}{\mathrm{d}}

\newcommand{\DD}{\mathrm{D}}
\newcommand{\KL}{\mathrm{KL}}
\newcommand{\ind}{\mathbbm{1}}

\newcommand{\VV}{\mathrm{V}}

\newcommand{\overbar}[1]{\mkern 1.5mu\overline{\mkern-1.5mu#1\mkern-1.5mu}\mkern 1.5mu}
\usepackage[utf8]{inputenc}
\usepackage{unicode_character}

\DeclarePairedDelimiter\floor{\lfloor}{\rfloor}
\DeclareMathOperator*{\argmin}{\arg\!\min}

\title{The iterative persuasion-polarization opinion dynamics and its mean-field analysis}

\author{Fei Cao \footnotemark[1] \and Stephanie Reed \footnotemark[2]}

\begin{document}
\maketitle

\footnotetext[1]{University of Massachusetts Amherst - Department of Mathematics and Statistics, 710 N Pleasant St, Amherst, MA 01003, USA}
\footnotetext[2]{California State University, Fullerton - Department of Mathematics, 800 N State College Blvd, Fullerton, CA 92831, USA}

\tableofcontents

\begin{abstract}
In this paper, we introduce the Iterative Persuasion-Polarization (IPP) model to study the dynamics of opinion formation and change within a population. The IPP model integrates mechanisms of persuasion and repulsion, where individuals influence each other through interactions that can either align opinions incrementally or lead to greater divergence. The probability of each interaction type is governed by a parameter $\alpha$, representing the population's receptiveness to persuasion. We investigate how these interaction dynamics shape the long-term distribution of opinions, examining conditions that promote consensus or polarization. By deriving a system of nonlinear and autonomous ordinary differential equations (ODEs), we provide a rigorous mathematical framework for analyzing the distributional behavior of opinions in large populations. Our findings contribute to a deeper understanding of social influence dynamics and their implications in complex social systems.
\end{abstract}



\section{Introduction}
\setcounter{equation}{0}

Decoding how people form opinions and how individuals within a population influence the opinions of others has been a subject of mathematical interest since at least the mid nineteen-sixties \cite{ableson_1964,ableson_1967}. In recent decades, the application of physics principles to various fields such as sociology and economics to describe natural phenomena has proliferated. The fields of sociophysics and econophysics \cite{cao_biased_2023,cao_binomial_2022,cao_entropy_2021,cao_interacting_2022,cao_quantitative_2024,cao_uniform_2024,lanchier_rigorous_2018, lanchier_rigorous_2018-1, lanchier_role_2018,lanchier_reed_2024} have witnessed significant developments by virtue of tools and techniques from statistical physics. In particular, sociophyics, introduced in \cite{galam_gefen_shapir_1982}, which applies statistical physics to social contexts to infer the inner workings of human social behavior. Sociophysics has seen tremendous growth in interest since the turn of the century, with many now-classical models being introduced and studied during this time \cite{axelrod_1997,sznajd_sznajd_2000,deffuant_neau_amblard_weisbuch_2000}. We now recall a series of prototypical sociophysics models which have been subject to intensive research activities over the last few decades: \\


\noindent\textbf{Axelrod Model:} Proposed by Robert Axelrod in \cite{axelrod_1997}, this model features agents with multiple attributes that form their opinion profiles. In each iteration, an agent $x$ and one of their neighbors $y$ are chosen. With probability equal to their cultural similarity they will interact. Through this interaction a feature for which $x$ and $y$ differ will be chosen at random and agent $x$ will update this feature to match that of agent $y$. The Axelrod model is further investigated in \cite{lanchier_shweinsberg_2012}.
\vspace{.2in}

\noindent\textbf{Sznajd Model:} Introduced in \cite{sznajd_sznajd_2000} and known as the ``United we stand, divided we fall'' (USDF) model, it was motivated by the Ising Model \cite{ising_1925} which was proposed in order to describe magnestism of atoms in matter. In the original model, the one dimensional lattice of length $N$ makes up the underlying social structure. Each individual has an binary opinion in the set $\{1,-1\}$. Letting $S_i$ be the opinion of agent $i=1,2,\ldots,N$, at each iteration a deterministic process based on the current configuration takes place;
\begin{itemize}
\item If $S_i\,S_{i+1} = 1$, then set $S_{i-1} = S_{i+2} = S_i$,
\item If $S_i\,S_{i+1} = -1$, then set $S_{i-1} = S_{i+1}$ and $S_{i+2} = S_i$.
\end{itemize}
Many variations of the USDF model have been studied in \cite{sznajd_tabiszewski_timpanaro_2011,ochrombel_2001, fraia_tosin_2020,loy_raviola_tosin_2022,slanina_lavicka_2003,stauffer_deoliveira_2002,sznajd_sznajd_weron_2021}.
\vspace{.2in}

\noindent\textbf{Bounded Confidence Models:} Bounded Confidence (BC) models consist of a class of models where neighbors influence each other only if their opinions do not differ too much. Classic BC models were proposed in \cite{deffuant_neau_amblard_weisbuch_2000,hegselmann_krause_2002,dittmer_2001,deffuant_amblard_weisbuch_faure_2002, krause_2000,bennaim_2005}. The Deffuant model \cite{deffuant_neau_amblard_weisbuch_2000} is one such pinoneering BC model that has been widely studied \cite{lanchier_li_2020, lanchier_2012,lanchier_mercer_2023,stauffer_meyerortmanns_2004, weisbuch_deffuant_amblard_nadal_2002} since its inception. The Deffuant model assumes that each agent $i$ has a continuous opinion $x_i$. The convergence parameter $\mu$ typically in $[0, 0.5]$ controls how large of a change is made when opinions are updated and $d$ determines the maximum allowable difference in opinion of two interacting agents. At each time step, two agents with opinions $x$ and $x_\star$ are chosen to interact. As long as $|x-x_\star|<d$, they update their opinions to $x'$ and $x_\star'$ respectively according to the following rule:
 \begin{equation*}
 \left\{\begin{aligned}
   x' & = x + \mu\,(x_\star-x), \\
   x_\star' & = x_\star + \mu\,(x-x_\star).
   \end{aligned}\right.
 \end{equation*}

Each of these models, despite their relatively simple interaction schemes, can lead to very rich and complex behaviors in the long run.
They also aim to capture real elements of how humans influence one another in reality. In this paper, we propose a novel model, the Iterative Persuasion-Polarization (IPP) Model, which highlights several key features of human interaction:

\begin{itemize}
\item We assume that not all interactions result in the ``persuasion'' of one of the interacting individuals although this is still a possibility. In real-world scenarios, it is also possible to have negative interactions resulting in ``repulsion'' (or more polarization than existed prior to the interaction).
\item We also aim to capture the idea that people typically do not just adopt the opinion of their neighbor through a single interaction, instead it is more natural that over the course of many interactions can one individual begin to shift their opinion in incremental stages.
\item Like the Deffuant model, we do not assume that ones opinion on a particular topic can be adequately captured with just a binary set like $\{-1,1\}$ representing two extremes. Instead we assume that each individual has an opinion existing on a broader spectrum to better capture the nuances of the sentiment of individuals.
\end{itemize}

Our model seeks to capture the nuanced mechanisms of persuasion and repulsion through gradual opinion changes during interactions. By examining these dynamics, we aim to provide insights into how a population's openness to persuasion impacts the long-term distribution of opinions, and under what conditions consensus or polarization occurs. The next section is dedicated to rigorously defining the model at hand and deriving a system of nonlinear and autonomous ordinary differential equations (ODEs) to describe its (distributional) behavior in the mean-field region when the number of agents becomes sufficiently large.


\section{Iterative persuasion-polarization (IPP) model}\label{sec:sec2}
\setcounter{equation}{0}


Consider a population of individuals of size $N$. At any given time, each individual is characterized by her opinion (on a given topic) which is an element in the discrete set of admissible opinions
\begin{equation}\label{eq:opinion_space}
K = \{-k,-k+1,\ldots,k-1,k\}.
\end{equation}

The dynamics of the IPP model are as follows:

\begin{itemize}
\item At a constant rate, pairs of individuals $(x,y)$ interact within the population. In each interaction, $x$ acts as the ``persuader'' and $y$ as the ``persuaded'', indicating that  $x$ attempts to influence $y$'s opinion.
\item With probability $\alpha$, the interaction is of the ``persuade'' type, and with probability $\beta = 1-\alpha$, the interaction is of the ``repel'' type.
\item If the interaction is of the ``persuade'' type, and $x$ and $y$ do not already share the same opinion, $y$ will adjust her opinion to be one unit closer to $x$'s opinion (i.e. the model moves towards \emph{consensus}).
\item If the interaction is of the ``repel'' type, and  $y$ does not already hold the most extreme opinion, she will adjust her opinion to be one unit further from $x$'s opinion (i.e. the model moves towards \emph{polarization}).
\item In the case where $x$ and $y$ already share the same opinion, no update is made.
\end{itemize}

We assume that $0\leq\alpha\leq 1$ is a measure of the receptiveness or openness of the population to persuasion. We can interpret $\alpha \ll 1$ to indicate that the population is very antagonistic where most interactions result in a polarizing result. On the other hand, $\alpha \approx 1$ indicates that the population is highly receptive and thus most interactions lead to persuasion.

%

Letting $v_\star$ and $v$ denote the pre-interaction opinions of $x$ and $y$ respectively, and letting $v_\star'$ and $v'$ denote the post-interaction opinions, at the microscopic level, we can describe an interaction in our model by the following rules:

 \begin{align}\label{eqn:interaction}
    v' & =
    \begin{cases}
        v+a & \text{if } v = -k < v_\star,\\
        v+a-b & \text{if } -k<v<v_\star,\\
        v-a+b & \text{if } k>v>v_\star,\\
        v-a & \text{if } v=k> v_\star;\\
    \end{cases}
   \\
    v_\star' & = v_\star.
    \end{align}

We now introduce a set of notations and terminologies used throughout the present paper. If $f$ is the law of a real-valued random variable $X$, and $\varphi$ is a generic (smooth) test function, we will denote the expected value of $\varphi(X)$ to be

\begin{align}\label{eqn:expectation}
    \langle\varphi(X)\rangle & =
    \begin{cases}
        \int_{\mathbb{R}}\varphi(x)\, f(x)\,\dd x & \text{if $X$ is a continuous random variable},\\
        \sum\limits_{x:f(x)>0} \varphi(x)\,f(x) & \text{if $X$ is a discrete random variable}.
    \end{cases}
    \end{align}


Having origins in kinetic theory, the original Boltzmann equation is a partial integro-differential equation meant to describe the particle density of (dilute) gases \cite{villani_2002_review}. As of the early 2000s, the Boltzmann equation has been a popular tool for studying interacting particle systems with applications in economics, sociology and biology \cite{toscani_2006, cordier_pareschi_toscani_2005, kashdan_pareschi_2012}. A clear analogy can be made between a gas composed of colliding molecules resulting in velocity changes on a microscopic level and a population of agents whose interactions can result in a change of opinion or exchange of wealth in the cases of sociology and economics respectively. For a more detailed historical account of the Boltzmann equation and its various applications to interacting particle systems, we refer the reader to \cite{pareschi_toscani_2013}. The weak form of the Boltzmann-type equation provided in \cite{fraia_tosin_2020} given by
\begin{equation}\label{eqn:boltzmann_cont}
\begin{aligned}
&\frac{\dd}{\dd t}\int_{\VV}\varphi(v)f(t,v)\,\,\dd v  \\
&~~ =\frac{1}{2}\left\langle\int_{\VV}\int_{\VV}B(v,v_\star)\left(\varphi(v')+\varphi(v_\star')-\varphi(v)-\varphi(v_\star)\right)f(t,v)\,f(t,v_\star)\,\dd v\,\dd v_\star\right\rangle,
\end{aligned}
\end{equation}
where $B = B(v,v_\star)$ is the rate of interaction, $f(v,t)$ is the distribution of opinions $v\in V$ (a set of admissible opinions) at time $t\ge0$, and $\varphi$ is an arbitrary test function, is often used as a starting point to describe various opinion dynamics models in the mean-field region. However, the set $K$ of opinions in our model is discrete and so as the basis for our model we will use
\begin{equation}\label{eqn:boltzmann_disc}
\begin{aligned}
&\frac{\dd}{\dd t}\sum\limits_{v \in K}\varphi(v)f(t,v) \\
&~~= \frac{1}{2}\left\langle\sum\limits_{v\in K}\sum\limits_{v_\star \in K}B(v,v_\star)\left(\varphi(v')+\varphi(v_\star')-\varphi(v)-\varphi(v_\star)\right)f(t,v)\,f(t,v_\star)\right\rangle.
\end{aligned}
\end{equation}
In broader models, it can be assumed that the interaction rate between pairs of individuals, denoted as $B = B(v,v_\star)$, may vary based on the opinions of the individuals involved. For convenience, we shall assume individuals always interact at twice the unit rate thus rendering $B \equiv 2$. We are ultimately interested the fraction of individuals with opinion $i$ at time $t$, and we shall denote this quantity to be $p_i(t)$. This allows us the write $f(t,v)$ in terms of $p_i(t)$ for $i\in K$ by

\begin{equation}\label{eqn:f}
f(t,v) = \sum_{i=-k}^k p_i(t)\,\ind\{v=i\}.
\end{equation}

%
%

    Given our interaction rules prescribed via \eqref{eqn:interaction}, it follows that for any test function $\varphi$ on $\VV$, we must have that

        \begin{equation}\label{eqn:varphi}\varphi(v') =
    \begin{cases}
        \varphi(v)\,b+\varphi(v+1)\,a & \text{if } v = -k < v_\star,\\
        \varphi(v-1)\,b+\varphi(v+1)\,a & \text{if } -k<v<v_\star,\\
        \varphi(v-1)\,a+\varphi(v+1)\,b & \text{if } k>v>v_\star,\\
        \varphi(v)\,b+\varphi(v-1)\,a & \text{if } v=k> v_\star.\\
    \end{cases}\end{equation}

    Combining \eqref{eqn:f} and \eqref{eqn:varphi} yields that

\begin{equation}\label{eqn:1}
\begin{aligned}
&\frac{\dd}{\dd t}\sum_{i=-k}^k\varphi(i)\,p_i \\
&=\left\langle\sum_{v_\star=-k}^k\sum_{v=-k}^k\left(\varphi(v')+\varphi(v_\star')-\varphi(v)-\varphi(v_\star)\right)p_v\,p_{v_\star}\right\rangle\\
&= \left\langle\sum_{v_\star=-k+1}^k a\left(\varphi(-k+1)-\varphi(-k)\right)p_{v_\star}\,p_{-k}\right\rangle + \left\langle\sum_{v_\star=-k}^{k-1} a\left(\varphi(k-1)-\varphi(k)\right)p_{v_\star}\,p_{k}\right\rangle \\
&\quad + \left\langle\sum_{v_\star=-k}^{k}\sum_{v=-k+1}^{v_\star-1} \left(b\,\varphi(v-1)+a\,\varphi(v+1)-\varphi(v)\right)p_v\,p_{v_\star}\right\rangle \\
&\quad + \left\langle\sum_{v_\star=-k}^{k}\sum_{v=v_\star+1}^{k-1} \left(a\,\varphi(v-1)+b\,\varphi(v+1)-\varphi(v)\right)p_v\,p_{v_\star}\right\rangle \\
&= \sum_{v_\star=-k+1}^k \alpha\,\left(\varphi(-k+1)-\varphi(-k)\right)p_{v_\star}\,p_{-k} + \sum_{v_\star=-k}^{k-1} \alpha\,\left(\varphi(k-1)-\varphi(k)\right)p_{v_\star}\,p_k \\
&\quad + \sum_{v_\star=-k}^{k}\sum_{v=-k+1}^{v_\star-1} \left(\beta\,\varphi(v-1)+\alpha\,\varphi(v+1)-\varphi(v)\right)p_v\,p_{v_\star} \\
&\quad + \sum_{v_\star=-k}^{k}\sum_{v=v_\star+1}^{k-1} \left(\alpha\,\varphi(v-1)+\beta\,\varphi(v+1)-\varphi(v)\right)p_v\,p_{v_\star}. \\
\end{aligned}
\end{equation}

Now for each $v\in K$ we take $\varphi_v(i) = \ind\{i=v\}$ and then insert $\varphi_v$ in place of $\varphi$ in \eqref{eqn:1} which
gives rise to the following system of nonlinear Boltzmann-type ODEs:

\begin{equation}\label{eqn:sys}
\left\{
\begin{aligned}
p'_{-k} & = \alpha\, p_{-k}\,p_{-k+1} + \beta\,(1-p_{-k}-p_{-k+1})\,p_{-k+1} - \alpha\,p_{-k}\,(1-p_{-k}),\\
p'_i & = p_{i-1}\left(\alpha\sum_{j=i}^{k} p_j+\beta\sum_{j=-k}^{i-2}p_j\right) +p_{i+1}\left(\alpha\sum_{j=-k}^{i} p_j +\beta\sum_{j=i+2}^{k}p_j\right) - p_i\,(1-p_i), ~~-k<i<k \\
p'_k & = \alpha\,p_{k}\,p_{k-1} + \beta\,(1-p_{k}-p_{k-1})\,p_{k-1} - \alpha\,p_k\,(1-p_k).\\
\end{aligned}\right.
\end{equation}


In the next section, we provide an analysis of the ODE model \eqref{eqn:sys} in the cases where $\alpha=1$, $0<\alpha<1$ and $\alpha=0$.

\section{Large time analysis}\label{sec:sec3}
\setcounter{equation}{0}

In this section, we take on the task of analyzing the long time behavior of solutions to the nonlinear ODE system \eqref{eqn:sys}. First, we perform a harmless relabeling, i.e., we will set
\begin{equation}\label{eq:relabling}
q_n = p_{-k+n} \quad \text{for all $0 \leq n \leq 2k$}
\end{equation}
and thus identifying ${\bf p} \coloneqq (p_{-k},p_{-k+1},\ldots,p_{k-1},p_k)$ with ${\bf q} \coloneqq (q_0,q_1,\ldots,q_{2k-1},q_{2k})$. In a nutshell, we simply shift the space of admissible opinions $K$ \eqref{eq:opinion_space} from $\{-k,-k+1,\ldots,k-1,k\}$ to $\{0,1,\ldots,2k-1,2k\}$ so that all allowable values of opinions belong to $\mathbb N$. After such simple relabeling of the solution vector and shifting of the opinion space, the ODE system \eqref{eqn:sys} now reads as

\begin{equation}\label{eqn:ODE_main}
\left\{
\begin{aligned}
q'_0 & = \alpha\, q_0\,q_1 + \beta\,(1-q_0-q_1)\,q_1 - \alpha\,q_0\,(1-q_0)\\
q'_n & = q_{n-1}\left(\alpha\sum_{j=n}^{2k} q_j+\beta\sum_{j=0}^{n-2}q_j\right)+ q_{n+1}\left(\alpha\sum_{j=0}^{n} q_j +\beta\sum_{j=n+2}^{2k} q_j\right) - q_n\,(1-q_n), ~~0<n<2k \\
q'_{2k} & = \alpha\,q_{2k}\,q_{2k-1} + \beta\,(1-q_{2k}-q_{2k-1})\,q_{2k-1} - \alpha\,q_{2k}\,(1-q_{2k})
\end{aligned}\right.
\end{equation}

We split our results on the large time convergence behavior of the solution to \eqref{eqn:ODE_main} into several subsections as the analysis depends heavily on the choice of the parameter $\alpha$ (or equivalently $\beta = 1-\alpha$) within the unit interval $[0,1]$, which measures the openness/persuasiveness of the entire society. In section \ref{subsec:3.1} we prove that for $\alpha = 1$, the solution of \eqref{eqn:ODE_main} converges to a two-point Bernoulli-type distribution which implies that only two nearby opinions around the average opinion survive in the long run. Section \ref{subsec:3.2} is devoted to the analysis of the system \eqref{eqn:ODE_main} when $\alpha = 0$, in which only two extreme opinions (represented by $0$ and $2k$) remain after large times. Lastly, we show in section \ref{subsec:3.3} that the solution to \eqref{eqn:ODE_main} converges to a uniform distribution on $\{0,1,\ldots,2k-1,2k\}$ when $\alpha = 1/2$. We emphasize that our main tool for the analysis of the long time behavior of the system \eqref{eqn:ODE_main} relies on the careful design of a $\alpha$-dependent Lyapunov functional, whose choice and design originate from our physical intuition regarding the ODE dynamics.

\subsection{Convergence to ``almost consensus'' for $\alpha =1$}\label{subsec:3.1}

When $\alpha = 1$ and hence $\beta = 0$, the nonlinear ODE system \eqref{eqn:ODE_main} boils down to
\begin{equation}\label{eqn:ODE_alpha=1}
\left\{
\begin{aligned}
q'_0 & = q_0\,q_1 - q_0\,(1-q_0)\\
q'_n & = q_{n-1}\,\sum_{j=n}^{2k} q_j + q_{n+1}\,\sum_{j=0}^{n} q_j  - q_n\,(1-q_n), ~~0<n<2k \\
q'_{2k} & = q_{2k}\,q_{2k-1} - q_{2k}\,(1-q_{2k})
\end{aligned}\right.
\end{equation}
With the convention that $q_{-1} \equiv 0$ and $q_{2k+1} \equiv 0$, we can recast the system \eqref{eqn:ODE_alpha=1} into a more compact form
\begin{equation}\label{eqn:ODE_alpha=1_compact}
q'_n = q_{n-1}\,\sum_{j=n}^{2k} q_j + q_{n+1}\,\sum_{j=0}^{n} q_j  - q_n\,(1-q_n)
\end{equation}
holding for all $0\leq n\leq 2k$. We now encapsulate several elementary observations regarding the solution of \eqref{eqn:ODE_alpha=1} into the following lemma.

\begin{lemma}\label{lem:1}
Assume that ${\bf q} = (q_0,q_1,\ldots,q_{2k-1},q_{2k})$ is a classical solution to the nonlinear system of ODEs \eqref{eqn:ODE_alpha=1} with ${\bf q}(t=0) \in \mathcal{P}(\{0,1,\ldots,2k\})$, and denote $\mu \coloneqq \sum_{n=0}^{2k} n\,q_n(0) \in [0,2k]$. Then we have
\begin{equation*}
{\bf q}(t) \in \mathcal{P}(\{0,1,\ldots,2k\}) \quad \textrm{and} \quad \sum_{n=0}^{2k} n\,q_n(t) = \mu
\end{equation*}
for all $t \geq 0$. Moreover, the unique equilibrium distribution ${\bf q}^* = (q^*_0,q^*_1,\ldots,q^*_{2k-1},q^*_{2k})$ associated to the ODE dynamics \eqref{eqn:ODE_alpha=1} is given by
\begin{equation}\label{eq:equil_alpha=1}
q^*_{\floor*{\mu}} = 1-\mu + \floor*{\mu},~~~q^*_{\floor*{\mu}+1} = \mu - \floor*{\mu},~~~\text{and}~~~ q^*_n = 0 ~~\text{for $n \notin \{\floor*{\mu},1+\floor*{\mu}\}$},
\end{equation}
in which $\floor*{\mu}$ represents the integer part of $\mu$.
\end{lemma}

The proof of this elementary lemma can be found in a very recent work \cite{perez_reed_2023} and hence will be omitted here. It is also worth mentioning that the authors of \cite{perez_reed_2023} also established a qualitative pointwise convergence result of the form ${\bf q}(t) \xrightarrow{t\to \infty} {\bf q}^*$ in the very special case when $\alpha = 1$ and $k=1$, in which scenario the ODE system \eqref{eqn:ODE_alpha=1} becomes explicit solvable. Our main goal lies in the designation of a suitable Lyapunov functional in order to capture certain quantitative aspects of the solution trajectory, for purpose we recall the definition of the so-called Gini index:

\begin{definition}[\textbf{Gini index}]
For a given probability mass function ${\bf q} \in \mathcal{P}(\mathbb N)$ with mean $\mu \in \mathbb{R}_+$, the Gini index of the distribution ${\bf q}$ (whose value belongs to $[0,1]$) is defined by
\begin{equation}\label{def1:Gini}
G[{\bf q}] = \frac{1}{2\,\mu} \sum\limits_{i\in \mathbb N}\sum\limits_{j \in \mathbb N} |i-j|\,q_i\,q_j.
\end{equation}
\end{definition}

The Gini index $G$ is a widely used concept in socio-economical context which serves as a measurement of (wealth) distributional inequality among a given society and ranges from $0$ to $1$. We will prove that the Gini index is a Lyapunov functional along the solution of the system \eqref{eqn:ODE_alpha=1} for all $t\geq 0$, the main motivation behind the choice of the Gini index as appropriate Lyapunov functional for the evolution \eqref{eqn:ODE_alpha=1} resides in the following variational characterization of the Bernoulli-type equilibrium distribution ${\bf q}^*$ \eqref{eq:equil_alpha=1}:

\begin{lemma}\label{lem:variational_characterization}
The Gini index is minimized at ${\bf q}^*$ among probabilities on $\{0,1,\ldots,2k\}$ with fixed mean value $\mu \in [0,2k]$. In other words, let
\begin{equation}\label{eq:space_of_probabilities}
\mathcal{S}_\mu \coloneqq \left\{{\bf q} \in \mathcal{P}(\{0,1,\ldots,2k\}) \mid \sum_{n=0}^{2k} n\,q_n = \mu \right\}.
\end{equation}
Then \begin{equation}\label{eq:p*_characterization}
{\bf q}^* = \argmin_{{\bf q} \in \mathcal{S}_\mu} G[{\bf q}].
\end{equation}
\end{lemma}

\begin{proof}
For the sake of notational simplicity, we introduce
\begin{equation}\label{eq:Gini_rescaled}
\tilde{G}[{\bf q}] \coloneqq \frac{1}{2}\sum\limits_{i=0}^{2k}\sum\limits_{j=0}^{2k} |i-j|\,q_i\,q_j
\end{equation}
as the re-scaled version of the Gini index. In other words, $\tilde{G}[{\bf q}] = \mu\,G[{\bf q}]$ where $\mu$ is the mean of the distribution ${\bf q}$. A straightforward computation yields that
\begin{equation*}
G[{\bf q}^*] = \frac{1}{\mu}\,q^*_{\floor*{\mu}}\,q^*_{\floor*{\mu}+1} = \frac{1}{\mu}\,(1-\mu + \floor*{\mu})\,(\mu - \floor*{\mu}),
\end{equation*}
hence $\tilde{G}[{\bf q}^*] = (1-\mu + \floor*{\mu})\,(\mu - \floor*{\mu})$. Our goal is to show that if ${\bf q} \in \mathcal{S}_\mu$ satisfies $q_m > 0$ for some $m \notin \{\floor*{\mu},1+\floor*{\mu}\}$, then $\tilde{G}[{\bf q}] > \tilde{G}[{\bf q}^*]$. Without loss of generality, we work with the scenario that $m \in \{0,1,\ldots,\floor*{\mu}-1\}$. We first prove the following preliminary result valid for all ${\bf q} \in \mathcal{S}_\mu$:
\begin{equation}\label{eq:preliminary}
\tilde{G}[{\bf q}] = \frac{1}{2}\sum\limits_{i=0}^{2k}\sum\limits_{j=0}^{2k} |i-j|\,q_i\,q_j \geq \sum\limits_{i=0}^{\floor*{\mu}} (\mu-i)\,q_i.
\end{equation}
Indeed, we have \begin{align*}
\tilde{G}[{\bf q}] &\geq \sum\limits_{i\leq\floor*{\mu}}\sum\limits_{j\geq \floor*{\mu}+1} (j-i)\,q_i\,q_j = \sum\limits_{j\geq \floor*{\mu}+1}\sum\limits_{i\leq\floor*{\mu}} (j-i)\,q_i\,q_j \\
&= \sum\limits_{i\leq\floor*{\mu}} q_i\cdot \sum\limits_{j\geq \floor*{\mu}+1} j\,q_j - \sum\limits_{i\leq\floor*{\mu}} i\,q_i\cdot \sum\limits_{j\geq \floor*{\mu}+1} q_j \\
&= \sum\limits_{i\leq\floor*{\mu}} q_i\cdot\left(\mu - \sum\limits_{j\leq\floor*{\mu}} j\,q_j\right)-\sum\limits_{i\leq\floor*{\mu}} i\,q_i\cdot \left(1 - \sum\limits_{j\leq\floor*{\mu} q_j}\right) \\
&= \mu\,\sum\limits_{0\leq i\leq\floor*{\mu}} q_i - \sum\limits_{i\leq\floor*{\mu}} i\,q_i = \sum\limits_{0\leq i\leq\floor*{\mu}} (\mu-i)\,q_i.
\end{align*}
Now it suffices to prove that $\sum_{0\leq i\leq\floor*{\mu}} (\mu-i)\,q_i > (1-\mu + \floor*{\mu})\,(\mu - \floor*{\mu})$. We divide the proof into two sub-cases depending on how large $q_m$ is (recall that $q_m > 0$ by our assumption).\\
\vspace{0.2in}
Case i): If $(\floor*{\mu}+1)(1-q_m) \leq \mu - m\,q_m$ or equivalently if $q_m \geq \frac{\floor*{\mu}+1-\mu}{\floor*{\mu}+1-m}$. Then
\begin{equation*}
\sum\limits_{0\leq i\leq\floor*{\mu}} (\mu-i)\,q_i \geq (\mu-m)\,q_m = (\floor*{\mu}+1-\mu)\,\frac{\mu-m}{\floor*{\mu}+1-m} > (\floor*{\mu}+1-\mu)\,(\mu - \floor*{\mu}),
\end{equation*}
where the last inequality follows from the fact that the function $x \mapsto \frac{\mu -x}{\floor*{\mu}+1-x}$ is strictly decreasing for all $x\in [0,\mu]$. \\
\vspace{0.2in}
Case ii): If $(\floor*{\mu}+1)(1-q_m) > \mu - m\,q_m$. In this case, there exist $m_1,m_2,\ldots,m_\ell \in \{0,1,\ldots,\floor*{\mu}\} \setminus \{m\}$ with $q_{m_i} > 0$ for all $1\leq i\leq \ell \leq \floor*{\mu}$ such that
\begin{equation*}
(\floor*{\mu}+1)(1-q_m-q_{m_1}-\cdots-q_{m_\ell}) \leq \mu - m\,q_m - m_1\,q_{m_1} - \cdots - m_\ell\,q_{m_\ell}.
\end{equation*}
Therefore, on the on hand,
\begin{equation}\label{eq:piece1}
\begin{aligned}
\sum\limits_{0\leq i\leq\floor*{\mu}} (\mu-i)\,q_i &\geq (\mu-m)\,q_m + (\mu-m_1)\,q_{m_1} + \cdots + (\mu-m_\ell)\,q_{m_\ell} \\
&> (\mu - \floor*{\mu})\,(q_m + q_{m_1} + \cdots + q_{m_\ell}).
\end{aligned}
\end{equation}
On the other hand, we also have
\begin{equation}\label{eq:piece2}
\begin{aligned}
\sum\limits_{0\leq i\leq\floor*{\mu}} (\mu-i)\,q_i &\geq (\mu-m)\,q_m + (\mu-m_1)\,q_{m_1} + \cdots + (\mu-m_\ell)\,q_{m_\ell} \\
&= \mu\,(q_m + q_{m_1} + \cdots + q_{m_\ell}) - (m\,q_m + m_1\,q_{m_1} + \cdots + m_\ell\,q_{m_\ell}) \\
&\geq \mu\,(q_m + q_{m_1} + \cdots + q_{m_\ell}) - \left[\mu - (\floor*{\mu}+1)(1-q_m-q_{m_1}-\cdots-q_{m_\ell})\right] \\
&= \floor*{\mu}+1-\mu + (\mu - \floor*{\mu}-1)\,(q_m + q_{m_1} + \cdots + q_{m_\ell})\\
&= (\floor*{\mu}+1-\mu)\,(1-q_m-q_{m_1}-\cdots-q_{m_\ell}).
\end{aligned}
\end{equation}
Assembling \eqref{eq:piece1} and \eqref{eq:piece2} together we deduce that
\[\sum\limits_{0\leq i\leq\floor*{\mu}} (\mu-i)\,q_i > (\mu - \floor*{\mu})\,(q_m + q_{m_1} + \cdots + q_{m_\ell}) \geq (\mu - \floor*{\mu})\,(\floor*{\mu}+1-\mu)\] if $q_m + q_{m_1} + \cdots + q_{m_\ell} \geq \floor*{\mu}+1-\mu$, and that \[\sum\limits_{0\leq i\leq\floor*{\mu}} (\mu-i)\,q_i \geq (\floor*{\mu}+1-\mu)\,(1-q_m-q_{m_1}-\cdots-q_{m_\ell}) > (\floor*{\mu}+1-\mu)\,(\mu - \floor*{\mu})\] if
$q_m + q_{m_1} + \cdots + q_{m_\ell} < \floor*{\mu}+1-\mu$.  Finally, we conclude that \[\tilde{G}[{\bf q}] \geq \sum_{i\leq \floor*{\mu}} (\mu-i)\,q_i > (\floor*{\mu}+1-\mu)\,(\mu - \floor*{\mu}) = \tilde{G}[{\bf q}^*]\] and the proof is completed.
\end{proof}

\begin{remark}
The content of Lemma \eqref{lem:variational_characterization} conveys a very clear intuition from a economic point of view: if we interpret $q_n$ as the fraction of individuals/agents with $n$ dollars in a closed economical system, where the average amount of dollars per agent equals to $\mu \in [0,2k]$, then heuristically it makes a perfect sense that the ``most egalitarian'' way of distributing a very large bulk of money among the agents (under the constraint that each agent must have integer-valued wealth ranging from $0$ to $2k$) is to set a proportion of $\floor*{\mu}+1-\mu$ agents to have $\floor*{\mu}$ dollars and a proportion of $\mu - \floor*{\mu}$ agents to have $\floor*{\mu}+1$ dollars. In fact, this economic intuition, partially inspired from many works in econophysics \cite{cao_derivation_2021,cao_explicit_2021,cao_sticky_2024,cao_uncovering_2022}, is the main reason that motivates us to perform the innocent shifting and relabeling \eqref{eq:relabling} at the beginning of this section.
\end{remark}

We are now ready to state the main convergence result in this section.
\begin{theorem}\label{thm:1}
For any $k \in \mathbb{N}_+$, if ${\bf q}(t)$ is a solution of the nonlinear system of ODEs \eqref{eqn:ODE_alpha=1} with ${\bf q}(0) \in \mathcal{S}_\mu$ and $\mu \in (0,2k)$, then for all $t\geq 0$ we have
\begin{equation}\label{eq:Gini_dissipation}
\frac{\dd}{\dd t} \tilde{G}[{\bf q}] = \mu\,\frac{\dd}{\dd t} G[{\bf q}] = -2\sum\limits_{0\leq i<j<\ell \leq 2k} q_i\,q_j\,q_\ell \leq 0.
\end{equation}
Consequently, the Gini index serves an Lyapunov functional along the solution trajectory of the system \eqref{eqn:ODE_alpha=1}, and ${\bf q}(t) \xrightarrow{t \to \infty} {\bf q}^*$.
\end{theorem}

\begin{proof}
We denote $F_{-1} = 0$ and $F_n = \sum_{i=0}^n q_i$ for $0\leq n \leq 2k$ as the cumulative distribution function associated to the probability mass function ${\bf q}$. Now we compute the time derivative of the re-scaled Gini index $\tilde{G}[{\bf q}]$ along the solution of \eqref{eqn:ODE_alpha=1} as follows:
\begin{align*}
\mu\,\frac{\dd}{\dd t} G[{\bf q}] &= \frac{\dd}{\dd t}\left[\frac 12\,\sum_{i,j=0}^{2k} |i-j|\,q_i\,q_j\right] = \frac{\dd}{\dd t}\sum_{i,j=0}^{2k} |i-j|\,q'_i\,q_j \\
&= \sum_{i,j=0}^{2k} |i-j|\,q_j\,\left[q_{i-1}\,\sum_{\ell=i}^{2k} q_\ell + q_{i+1}\,\sum_{\ell=0}^i q_\ell  - q_i\,(1-q_i)\right] \\
&= \sum_{i=0}^{2k}\sum_{j=0}^{2k} |i-j|q_{i-1}\,\sum_{\ell=i}^{2k} q_\ell\,q_j + \sum_{i=0}^{2k}\sum_{j=0}^{2k}|i-j|q_{i+1}\,\sum_{\ell=0}^i q_\ell\,q_j - \sum_{i=0}^{2k}\sum_{j=0}^{2k} |i-j|q_i(1-q_i)q_j \\
&= \sum_{i=0}^{2k}\sum_{j=0}^{2k} \left\{|i+1-j|\,\sum_{\ell=i+1}^{2k} q_\ell + |i-1-j|\,\sum_{\ell=0}^{i-1} q_\ell - |i-j|\,(1-q_i)\right\}q_i\,q_j \\
&= \sum_{i=0}^{2k}\sum_{j=0}^{2k} \Big(|i+1-j|\,(1-q_i) + |i-1-j|\,\sum_{\ell=0}^{i-1} (|i-1-j|-|i+1-j|)\,q_\ell \\
&\qquad \qquad - |i-j|\,(1-q_i)\Big)q_i\,q_j \\
&= \sum_{i=0}^{2k}\sum_{j=0}^{2k} q_j\,q_i\,(1-q_i)\,(|i+1-j|-|i-j|) + \sum_{i=0}^{2k}\sum_{j=0}^{2k} q_j\,q_i\,(|i-1-j|-|i+1-j|)\,F_{i-1}\\
&= \sum_{i=0}^{2k} q_i\,(1-q_i)\,\left(\sum_{j\leq i} q_j - \sum_{j> i} q_j\right) - \sum_{i=0}^{2k} F_{i-1}\,q_i\,\left(-2\sum_{j\leq i-1} q_j + 2\sum_{j\geq  i+1} q_j\right) \\
&= 2\,\sum_{i=0}^{2k} q_i\,(1-q_i)\,F_i - \sum_{i=0}^{2k} q_i\,(1-q_i) - \sum_{i=0}^{2k} F_{i-1}\,q_i\,\left(2 - 2\,F_{i-1} - 2\,F_i\right) \\
&= \sum_{i=0}^{2k} q_i\,\left[2\,(1-q_i)\,F_i - (1-q_i) + F_{i-1}\,(2-2\,q_i - 4\,F_{i-1})\right] \\
&= \sum_{i=0}^{2k} q_i\,\left[2\,(1-q_i)\,F_i - (1-q_i) + 2\,(1-q_i)\,F_{i-1} 4\,F^2_{i-1}\right] \\
&= \sum_{i=0}^{2k} q_i\,\left[4\,F_{i-1}\,(1-F_i) + 2\,(1-q_i)\,q_i - (1-q_i)\right].
\end{align*}
Now, we compute
\begin{equation*}
4\,\sum_{i=0}^{2k} q_i\,F_{i-1}\,(1-F_i) = 4\,\sum_{i=0}^{2k} q_i\,\sum_{j=0}^{i-1} q_j\,\sum_{\ell = i+1}^{2k} q_\ell = 4\sum\limits_{0\leq i<j<\ell \leq 2k} q_i\,q_j\,q_\ell
\end{equation*}
and notice that \begin{align*}
&\sum_{i=0}^{2k} \left(2\,(1-q_i)\,q^2_i - (1-q_i)\,q_i\right) = \sum_{i=0}^{2k} (3\,q^2_i - 2\,q^3_i) - \sum_{i=0}^{2k} q_i \\
&= \sum_{i=0}^{2k} (3 - 2\,q_i)\,q^2_i - \left(\sum_{i=0}^{2k} q_i\right)^2 \\
&= \sum_{i=0}^{2k} \left(1+ 2\,\sum_{j\neq i} q_j\right)\,q^2_i - \left(\sum_{i=0}^{2k} q^2_i + \sum_{i=0}^{2k}\sum_{j\neq i} q_j\,q_i\right) \\
&= \sum_{i=0}^{2k} q^2_i + 2\,\sum_{i=0}^{2k}\sum_{j\neq i} q_j\,q^2_i - \left(\sum_{i=0}^{2k} q^2_i + \sum_{i=0}^{2k}\sum_{j\neq i} q_j\,q_i\cdot \sum_{\ell=0}^k q_\ell\right) \\
&= 2\,\sum_{i=0}^{2k}\sum_{j\neq i} q_j\,q^2_i - \left(\sum_{i=0}^{2k}\sum_{j\neq i} q_j\,q^2_i + \sum_{i=0}^{2k} q_i\,\sum_{j\neq i} q_j\,\sum_{\ell \neq i} q_\ell\right) \\
&= 2\,\sum_{i=0}^{2k}\sum_{j\neq i} q_j\,q^2_i - \left(2\,\sum_{i=0}^{2k}\sum_{j\neq i} q_j\,q^2_i + \sum_{\substack{0\leq i,j,\ell\leq 2k\\i\neq j\neq \ell}} q_i\,q_j\,q_\ell \right) \\
&= -3!\sum\limits_{0\leq i<j<\ell \leq 2k} q_i\,q_j\,q_\ell = -6\sum\limits_{0\leq i<j<\ell \leq 2k} q_i\,q_j\,q_\ell.
\end{align*}
Therefore, we deduce that
\begin{align*}
\frac{\dd}{\dd t} \tilde{G}[{\bf q}] = \mu\,\frac{\dd}{\dd t} G[{\bf q}] &= 4\sum\limits_{0\leq i<j<\ell \leq 2k} q_i\,q_j\,q_\ell - 6\sum\limits_{0\leq i<j<\ell \leq 2k} q_i\,q_j\,q_\ell \leq 0\\
&= -2\sum\limits_{0\leq i<j<\ell \leq 2k} q_i\,q_j\,q_\ell
\end{align*}
and the desired (pointwise) convergence ${\bf q}(t) \xrightarrow{t \to \infty} {\bf q}^*$ follows from the variational characterization \eqref{eq:p*_characterization} (see \cite{boghosian_h_2015} for the use of a similar strategy employed here).
\end{proof}

To illustrate the dissipation of the Gini index numerically, we use $k=2$, ${\bf q}(t=0) = (0.25, 0.2, 0.35, 0.2, 0)$, and the standard Runge-Kutta fourth-order algorithm to solve the ODE system \eqref{eqn:ODE_alpha=1} with time step $\Delta t = 0.001$. We plot in figure \ref{fig:alpha=1,k=2}-left and figure \ref{fig:alpha=1,k=2}-right the evolution of $G[{\bf q}(t)] - G[{\bf q}^*]$ and the solution vector ${\bf q}(t)$ with respect to time.

\begin{figure}[!htb]
  \begin{subfigure}{0.45\textwidth}
    \centering
    \includegraphics[scale=0.55]{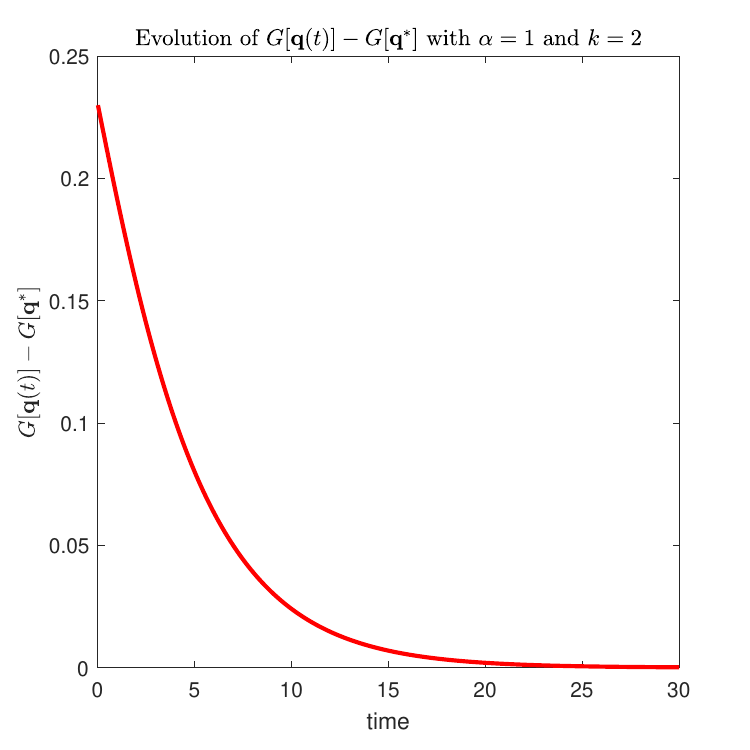}
  \end{subfigure}
  \hspace{0.1in}
  \begin{subfigure}{0.45\textwidth}
    \centering
    \includegraphics[scale=0.55]{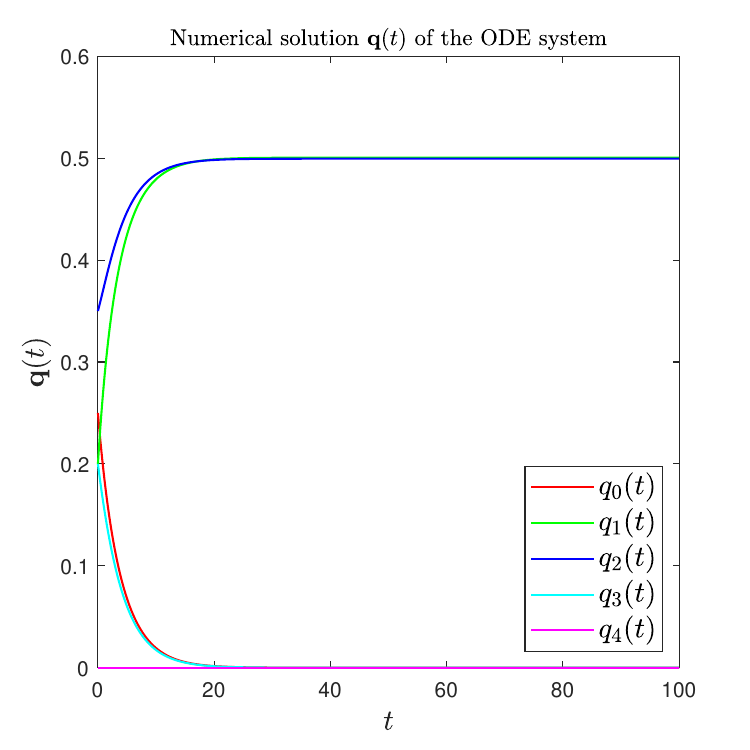}
  \end{subfigure}
  \caption{{\bf Left}: Decay of $G[{\bf q}(t)] - G[{\bf q}^*]$ along the solution of \eqref{eqn:ODE_alpha=1}. {\bf Right}: Evolution of the solution vector ${\bf q}(t)$ with respect to time.} 
  \label{fig:alpha=1,k=2}
\end{figure}

Although we managed to show that the monotonicity of the Gini index is the underlying mechanism which drives the solution of \eqref{eqn:ODE_alpha=1} to its unique equilibrium distribution ${\bf q}^*$, it is a very challenging task to search for a quantitative decay of the Gini index (along the solution of \eqref{eqn:ODE_alpha=1}) with respect to time, which amounts to establishing a explicit differential inequality satisfied by the time derivative of $G[{\bf q}]$. Fortunately, in the simplest case where $k=1$ we can indeed prove a quantitative bound on $G[{\bf q}]$, to which we now turn:

\begin{theorem}\label{thm:2}
If ${\bf q}(t)$ is a solution of the nonlinear system of ODEs \eqref{eqn:ODE_alpha=1} with $k=1$, ${\bf q}(0) \in \mathcal{S}_\mu$ and $\mu \in (0,2)$, then there exist some $\delta > 0$ and some explicitly computable $t^* > 0$ such that the following estimates are valid for all $t\geq t^*$:
\begin{equation}\label{eq:Gini_bound}
\tilde{G}[{\bf q}(t)] - \tilde{G}[{\bf q}^*] \leq \left\{\begin{aligned}
&\frac{1}{\frac{\delta}{2}\,(t-t^*)+\frac{1}{\tilde{G}[{\bf q}(t^*)]}} \qquad \textrm{if~} \mu = 1,\\
&\left(\tilde{G}[{\bf q}(t^*)] - \tilde{G}[{\bf q}^*]\right)\expo^{-\frac{|\mu-1|}{\min\{\mu,2-\mu\}}\,\delta\,(t-t^*)} \quad \text{if~} \mu \neq 1.
\end{aligned}
\right.
\end{equation}
\end{theorem}

\begin{proof}
In the special case $k=1$, the equilibrium distribution ${\bf q}^*$ boils down to
\[q^*_0 = \max\{1-\mu,0\},\quad  q^*_1 = \min\{2-\mu,\mu\}, \quad q^*_2 = \max\{\mu-1,0\}\]
and \begin{equation*}
\tilde{G}[{\bf q}^*] = (1-\mu + \floor*{\mu})\,(\mu - \floor*{\mu}) = \left\{\begin{aligned}
& \mu\,(1-\mu), \quad \textrm{if~} 0 < \mu \leq 1,\\
& 3\,\mu-2-\mu^2, \quad \textrm{if~} 1\leq \mu < 2.
\end{aligned}\right.
\end{equation*}
On the other hand, the re-scaled Gini index simplifies to \[\tilde{G}[{\bf q}] = \frac{1}{2}\sum\limits_{i=0}^{2}\sum\limits_{j=0}^{2} |i-j|\,q_i\,q_j = q_0\,q_1 + q_1\,q_2 + 2\,q_0\,q_2.\] Now since ${\bf q} \in \mathcal{S}_\mu$, $q_0 + q_1 + q_2 = 1$ and $q_1 + 2\,q_2 = \mu$, whence $q_2 = \frac{\mu - q_1}{2}$ and $q_0 = \frac{2-\mu-q_1}{2}$. Consequently, $q_1 \leq \min\{\mu,2-\mu\}$ and we can express the re-scaled Gini index $\tilde{G}[{\bf q}] = \mu\,G[{\bf q}]$ in terms of $q_1$ sorely:
\begin{equation*}
\begin{aligned}
\tilde{G}[{\bf q}] = q_0\,q_1 + q_1\,q_2 + 2\,q_0\,q_2 &= \frac{2-\mu-q_1}{2}\,q_1 + \frac{\mu - q_1}{2}\,q_1 + 2\,\frac{2-\mu-q_1}{2}\,\frac{\mu - q_1}{2} \\
&= \frac{2\,\mu - \mu^2 - q^2_1}{2}.
\end{aligned}
\end{equation*}
Thus we arrive at \[\tilde{G}[{\bf q}] - \tilde{G}[{\bf q}^*] = \frac{(\min\{\mu,2-\mu\})^2 - q^2_1}{2} \geq 0.\] In order to derive a differential inequality for $\tilde{G}[{\bf q}] - \tilde{G}[{\bf q}^*]$, the goal becomes bounding $q_0\,q_1\,q_2 = -\frac 12\,\frac{\dd}{\dd t} \tilde{G}[{\bf q}]$ from below by some function of $\tilde{G}[{\bf q}] - \tilde{G}[{\bf q}^*]$. We notice that in the case $k=1$, the ODE system \eqref{eqn:ODE_alpha=1_compact} implies that $q_1$ is increasing with respect to time, and since $q_1(t) \xrightarrow{t\to \infty} q^*_1 > 0$, for a small enough $\delta > 0$ (for instance, one may take $\delta = q^*_1 / 2$) we can always find some finite time $t^*$ (depending only on the initial datum and $\delta$) such that $q_1(t) \geq \delta$ for all $t\geq t^*$. In the sequel, all the differential inequalities we obtain below will be valid when time $t$ is larger than $\delta$. We divide the derivation of the relevant differential inequalities below into three sub-cases depending on the range of $\mu \in (0,2)$.\\
\vspace{0.2in}
Case i): If $0<\mu <1$. We have
\begin{align*}
-\frac{\dd}{\dd t} \tilde{G}[{\bf q}] &= -\frac{\dd}{\dd t} \left(\tilde{G}[{\bf q}] - \tilde{G}[{\bf q}^*]\right) = 2\,q_0\,q_1\,q_2 \\
&= 2\,q_1\,\frac{2-\mu-q_1}{2}\,\frac{\mu - q_1}{2} = \frac{q_1}{2}\,(2-\mu-q_1)\,(\mu - q_1) \geq \frac{q_1}{2}\,(2-2\,\mu)\,(\mu - q_1) \\
&= \frac{q_1}{2}\,\frac{2-2\,\mu}{2\,\mu}\,2\,\mu\,(\mu - q_1) \geq \frac{q_1}{2}\,\frac{2-2\,\mu}{2\,\mu}\,(\mu-q_1)\,(\mu+q_1) \\
&\geq \delta\,\frac{2-2\,\mu}{2\,\mu}\,\left(G[{\bf q}] - G[{\bf q}^*]\right),
\end{align*}
from which Gronwall's lemma leads us to \[\tilde{G}[{\bf q}(t)] - \tilde{G}[{\bf q}^*] \leq \left(\tilde{G}[{\bf q}(t^*)] - \tilde{G}[{\bf q}^*]\right)\expo^{-\frac{1-\mu}{\mu}\,\delta\,(t-t^*)} \] for all $t \geq t^*$. \\
\vspace{0.2in}
Case ii): If $1<\mu<2$. We use the fact that $q_1\leq 2-\mu$ to deduce
\begin{align*}
-\frac{\dd}{\dd t} \tilde{G}[{\bf q}] &= -\frac{\dd}{\dd t} \left(\tilde{G}[{\bf q}] - \tilde{G}[{\bf q}^*]\right) = \frac{q_1}{2}\,(2-\mu-q_1)\,(\mu - q_1) \\
&\geq \frac{q_1}{2}\,(2-\mu-q_1)\,(2\,\mu - 2) = \frac{q_1}{2}\,(2-\mu-q_1)\,\frac{2\,\mu-2}{2\,(2-\mu)}\,2\,(2-\mu) \\
&\geq \frac{q_1}{2}\,(2-\mu-q_1)\,\frac{2\,\mu-2}{2\,(2-\mu)}\,(2-\mu+q_1) \geq \frac{2\,\mu-2}{2\,(2-\mu)}\,\delta\,\left(\tilde{G}[{\bf q}] - \tilde{G}[{\bf q}^*]\right),
\end{align*}
from which Gronwall's lemma gives rise to \[\tilde{G}[{\bf q}(t)] - \tilde{G}[{\bf q}^*] \leq \left(\tilde{G}[{\bf q}(t^*)] - \tilde{G}[{\bf q}^*]\right)\expo^{-\frac{\mu-1}{2-\mu}\,\delta\,(t-t^*)} \] for all $t \geq t^*$. \\
\vspace{0.2in}
Case iii): If $\mu = 1$. Then on the one hand, \[-\frac{\dd}{\dd t} \left(\tilde{G}[{\bf q}] - \tilde{G}[{\bf q}^*]\right) = \frac{q_1}{2}\,(2-\mu-q_1)\,(\mu - q_1) = \frac{q_1}{2}\,(1-q_1)^2.\] On the other hand, we also have \[\left(\tilde{G}[{\bf q}] - \tilde{G}[{\bf q}^*]\right)^2 = \left(\frac{1-q^2_1}{2}\right)^2 = \frac{(1+q_1)^2}{4}\,(1-q_1)^2.\] As a result, $-\frac{\dd}{\dd t} \tilde{G}[{\bf q}] \geq \frac{\delta}{2}\,\left(\tilde{G}[{\bf q}] - \tilde{G}[{\bf q}^*]\right)^2$ for all $t\geq t^*$ and Gronwall's inequality yields that \[\tilde{G}[{\bf q}(t)] = \tilde{G}[{\bf q}(t)] - \tilde{G}[{\bf q}^*] \leq \frac{1}{\frac{\delta}{2}\,(t-t^*)+\frac{1}{\tilde{G}[{\bf q}(t^*)]}}\] for all $t\geq t^*$.
\end{proof}

To illustrate the quantitative convergence guarantees reported in Theorem \ref{thm:2} (with $k=1$). We use two sets of the initial datum: ${\bf q}^{(1)}(t=0) = (0.2, 0.3, 0.5)$ and ${\bf q}^{(2)}(t=0) = (0.5, 0, 0.5)$ respectively. We plot in figure \ref{fig:alpha=1,k=1,ex1}-left and figure \ref{fig:alpha=1,k=1,ex1}-right the evolution of $G[{\bf q}(t)] - G[{\bf q}^*]$ in the normal scale and the semi-logy scale, starting from ${\bf q}^{(1)}(t=0)$. Similarly, we show in figure \ref{fig:alpha=1,k=1,ex2}-left and figure \ref{fig:alpha=1,k=1,ex2}-right the evolution of $G[{\bf q}(t)] - G[{\bf q}^*]$ in the normal scale and the log-log scale, starting from ${\bf q}^{(2)}(t=0)$.

\begin{figure}[!htb]
  \begin{subfigure}{0.45\textwidth}
    \centering
    \includegraphics[scale=0.55]{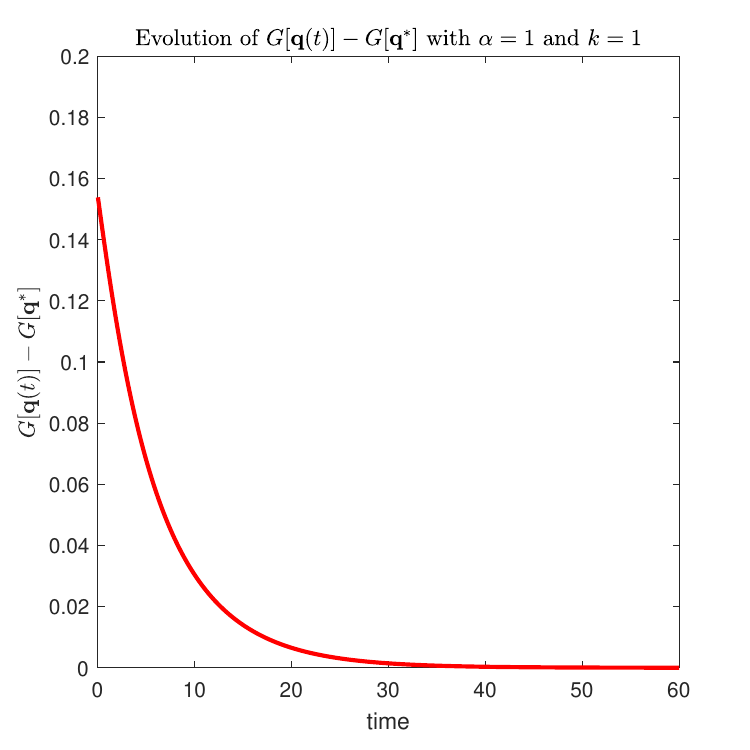}
  \end{subfigure}
  \hspace{0.1in}
  \begin{subfigure}{0.45\textwidth}
    \centering
    \includegraphics[scale=0.55]{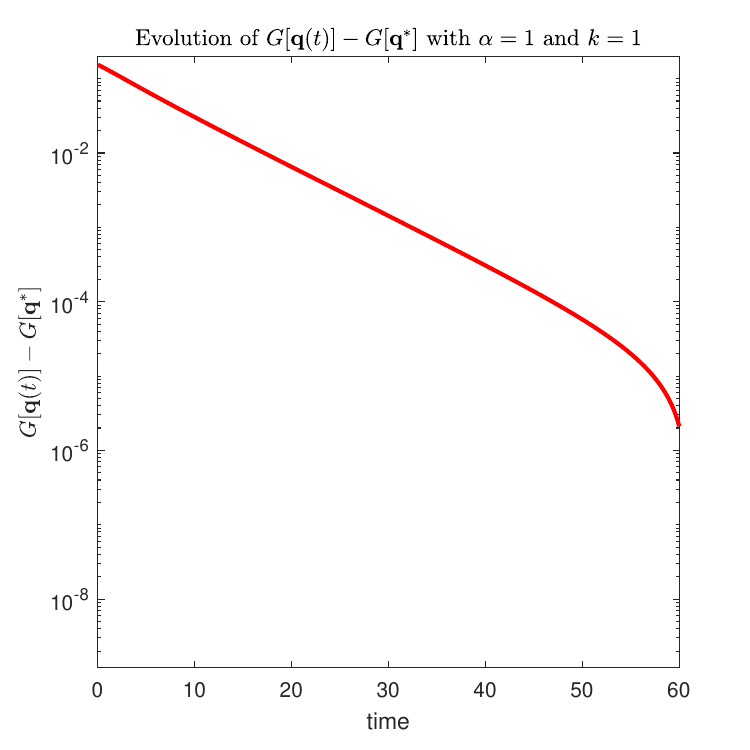}
  \end{subfigure}
  \caption{{\bf Left}: Decay of $G[{\bf q}(t)] - G[{\bf q}^*]$ along the solution of \eqref{eqn:ODE_alpha=1} with $k=1$ and ${\bf q}(t=0) = (0.2, 0.3, 0.5)$. {\bf Right}: Decay of $G[{\bf q}(t)] - G[{\bf q}^*]$ in the semi-logy scale. The decay is exponentially fast with respect to time, as predicted by Theorem \ref{thm:2}.}
  \label{fig:alpha=1,k=1,ex1}
\end{figure}

\begin{figure}[!htb]
  \begin{subfigure}{0.45\textwidth}
    \centering
    \includegraphics[scale=0.55]{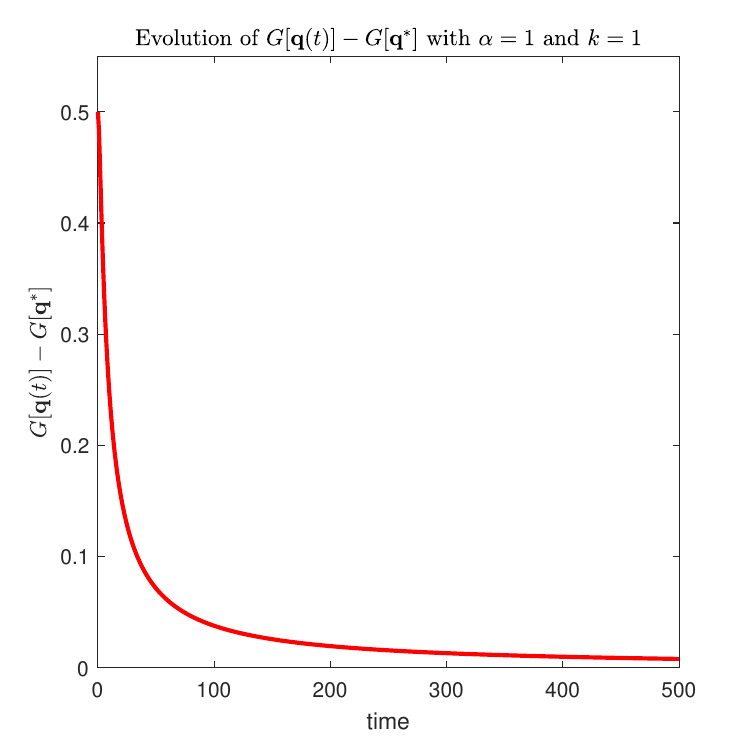}
  \end{subfigure}
  \hspace{0.1in}
  \begin{subfigure}{0.45\textwidth}
    \centering
    \includegraphics[scale=0.55]{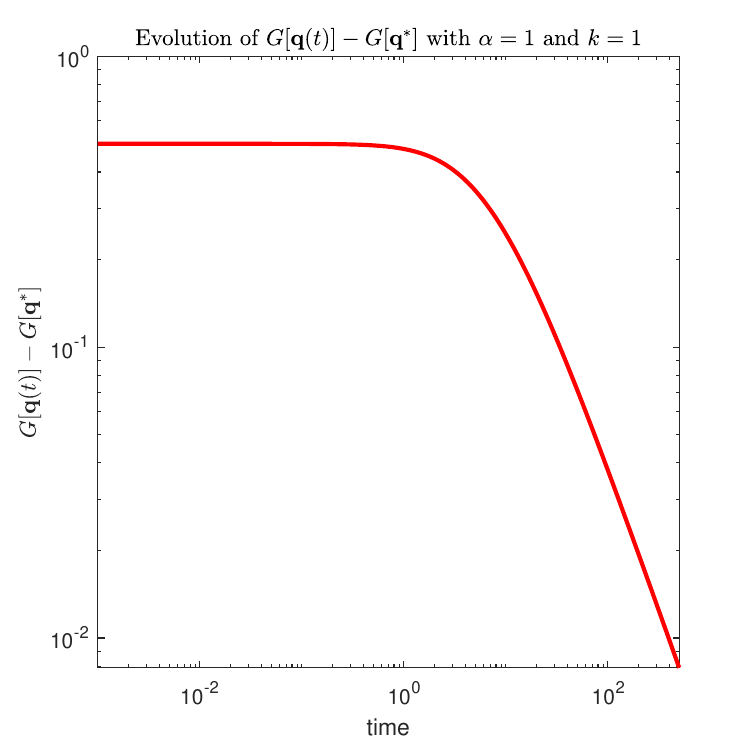}
  \end{subfigure}
  \caption{{\bf Left}: Decay of $G[{\bf q}(t)] - G[{\bf q}^*]$ along the solution of \eqref{eqn:ODE_alpha=1} with $k=1$ and ${\bf q}(t=0) = (0.5, 0, 0.5)$. {\bf Right}: Decay of $G[{\bf q}(t)] - G[{\bf q}^*]$ in the log-log scale. The decay is inversely proportional to time, as justified by Theorem \ref{thm:2}.}
  \label{fig:alpha=1,k=1,ex2}
\end{figure}

\subsection{Emergence of polarized society for $\alpha = 0$}\label{subsec:3.2}

In the case where $\alpha = 0$ or $\beta = 1$, the nonlinear ODE system \eqref{eqn:ODE_main} simplifies to
\begin{equation}\label{eqn:ODE_alpha=0}
\left\{
\begin{aligned}
q'_0 & = (1-q_0-q_1)\,q_1\\
q'_n & = q_{n-1}\,\sum_{j=0}^{n-2} q_j + q_{n+1}\,\sum_{j=n+2}^{2k} q_j  - q_n\,(1-q_n), ~~0<n<2k \\
q'_{2k} & = (1-q_{2k-1}-q_{2k})\,q_{2k-1}
\end{aligned}\right.
\end{equation}
We assume as usual that ${\bf q}(t=0) \in \mathcal{S}_\mu$ for some $\mu \in [0,2k]$, and observe that the evolution \eqref{eqn:ODE_alpha=0} preserves the total probability mass again since $\frac{\dd}{\dd t} \sum_{n=0}^{2k} q_n = 0$. However, the average opinion defined by $\mu(t) \coloneqq \sum_{n=0}^{2k} n\,q_n(t)$ will no longer be conserved at time evolves due to simple computation \[\frac{\dd}{\dd t} \mu(t) = \sum_{n=0}^{2k} n\,q'_n(t) = (q_0(t) - q_{2k}(t))\,(1-q_0(t) - q_{2k}(t)).\] Therefore, the re-scaled Gini index of the distribution ${\bf q}(t)$ now reads as
\begin{equation}\label{eq:Gini_rescaled_time_dependent}
\tilde{G}[{\bf q}(t)] = \mu(t)\,G[{\bf q}(t)] = \frac{1}{2}\,\sum\limits_{i=0}^{2k}\sum\limits_{j=0}^{2k} |i-j|\,q_i(t)\,q_j(t).
\end{equation}
Moreover, the dynamics admits a one-parameter family of equilibrium distributions supported only on two extreme opinions $0$ and $2k$:
\begin{equation}\label{eq:class}
\mathcal{A}_\gamma \coloneqq \left\{{\bf q} \in \mathcal{P}(\{0,1,\ldots,2k\}) \mid {\bf q} = \gamma\,\delta_0 + (1-\gamma)\,\delta_{2k},~~\gamma \in [0,1] \right\}.
\end{equation}
We now prove that the re-scaled Gini index still serves as a Lyapunov functional for the nonlinear system of ODEs \eqref{eqn:ODE_alpha=0}.
\begin{proposition}\label{prop:1}
For any $k \in \mathbb{N}_+$, if ${\bf q}(t)$ is a solution of the nonlinear system of ODEs \eqref{eqn:ODE_alpha=1} with ${\bf q}(0) \in \mathcal{S}_\mu$ and $\mu \in (0,2k)$, then for all $t\geq 0$ we have
\begin{equation}\label{eq:Gini_production}
\begin{aligned}
\frac{\dd}{\dd t} \tilde{G}[{\bf q}] &= (q_0-q_{2k})^2\,\sum\limits_{\ell=1}^{2k-1} q_\ell + \sum\limits_{\ell=1}^{2k-1} q^2_\ell\,(1-q_\ell) + \sum\limits_{\ell=1}^{2k-1} q^2_\ell\,(1-q_\ell-q_0-q_{2k}) \\
&\qquad  + 2\sum\limits_{1\leq i<j<\ell \leq 2k-1} q_i\,q_j\,q_\ell \\
&\geq 0.
\end{aligned}
\end{equation}
\end{proposition}

\begin{proof}
A similar lengthy computations as provided in the proof of Theorem \ref{thm:1} allow us to arrive at
\begin{equation*}
\frac{\dd}{\dd t} \tilde{G}[{\bf q}] = \sum_{i=0}^{2k} q_i\,\left[1+3\,q_i - 4\,F_i\,(1-F_{i-1})\right] - q_0\,(1-q_0) - q_{2k}\,(1-q_{2k}).
\end{equation*}
Now we recall that the proof of Theorem \ref{thm:1} also yields the following (generic) relation:
\begin{equation*}
\sum\limits_{i=0}^{2k} q_i\,4\,F_{i-1}\,(1-F_i) + \sum\limits_{i=0}^{2k} q_i\,(3\,q_i - 2\,q^2_i - 1) = -2\sum\limits_{0\leq i<j<\ell \leq 2k} q_i\,q_j\,q_\ell.
\end{equation*}
Since \begin{align*}
q_i\,4\,F_i\,(1-F_{i-1})&=4\,q_i\,(F_{i-1}+q_i)\,(1-F_i+q_i) \\
&= 4\,q_i\,F_{i-1}\,(1-F_i) + 4\,q_i\,\left(F_{i-1}\,q_i + q_i\,(1-F_i)+q^2_i\right) \\
&= 4\,q_i\,F_{i-1}\,(1-F_i) + 4\,q^2_i,
\end{align*}
we obtain \begin{align*}
\sum\limits_{i=0}^{2k} q_i\,4\,F_{i-1}\,(1-F_i) &= \sum\limits_{i=0}^{2k} 4\,q_i\,F_{i-1}\,(1-F_i) + \sum\limits_{i=0}^{2k} 4\,q^2_i \\
&= \sum\limits_{i=0}^{2k} 4\,q^2_i - 2\sum\limits_{0\leq i<j<\ell \leq 2k} q_i\,q_j\,q_\ell - \sum\limits_{i=0}^{2k} q_i\,(3\,q_i - 2\,q^2_i - 1).
\end{align*}
Therefore, \begin{equation*}
\frac{\dd}{\dd t} \tilde{G}[{\bf q}] = \sum\limits_{i=0}^{2k} 2\,q^2_i\,(1-q_i) + 2\sum\limits_{0\leq i<j<\ell \leq 2k} q_i\,q_j\,q_\ell - q_0\,(1-q_0) - q_{2k}\,(1-q_{2k}).
\end{equation*}
We also notice that \begin{align*}
q_0\,(1-q_0) &= q^2_0\,(1-q_0) + q_0\,(1-q_0)^2 = q^2_0\,(1-q_0) + q_0\,(q_1 + \cdots + q_{2k})^2 \\
&= q^2_0\,(1-q_0) + 2\,q_0\,\sum\limits_{1\leq i<j\leq 2k} q_i\,q_j + q_0\,(q^2_1+\cdots + q^2_{2k})
\end{align*}
and that \begin{align*}
q_{2k}\,(1-q_{2k}) &= q^2_{2k}\,(1-q_{2k}) + q_{2k}\,(1-q_{2k})^2 = q^2_{2k}\,(1-q_{2k}) + q_{2k}\,(q_0 + \cdots + q_{2k-1})^2 \\
&= q^2_{2k}\,(1-q_{2k}) + 2\,q_{2k}\,\sum\limits_{0\leq i<j\leq 2k-1} q_i\,q_j + q_{2k}\,(q^2_0+\cdots + q^2_{2k-1}).
\end{align*}
Finally, we conclude that
\begin{align*}
\frac{\dd}{\dd t} \tilde{G}[{\bf q}] &= 2\,\sum\limits_{\ell=1}^{2k-1} q^2_\ell\,(1-q_\ell) + 2\left[\sum\limits_{1\leq i<j<\ell \leq 2k-1} q_i\,q_j\,q_\ell -q_0\,q_{2k}\,\sum\limits_{1\leq \ell\leq 2k-1}q_\ell\right] \\
&\qquad + q^2_0\,(1-q_0) + q^2_{2k}\,(1-q_{2k}) - q_0\,(q^2_1+\cdots + q^2_{2k})-q_{2k}\,(q^2_0+\cdots + q^2_{2k-1}) \\
&= q^2_0\,q_{2k} + q^2_0\,\sum\limits_{\ell=1}^{2k-1} q_\ell + q^2_{2k}\,q_0 + q^2_{2k}\,\sum\limits_{\ell=1}^{2k-1} q_\ell - \sum\limits_{1\leq \ell\leq 2k-1} q^2_\ell\,(q_0+q_{2k}) - q^2_0\,q_{2k} - q^2_{2k}\,q_0 \\
&\quad +2\,\sum\limits_{\ell=1}^{2k-1} q^2_\ell\,(1-q_\ell) + 2\left[\sum\limits_{1\leq i<j<\ell \leq 2k-1} q_i\,q_j\,q_\ell -q_0\,q_{2k}\,\sum\limits_{1\leq \ell\leq 2k-1}q_\ell\right] \\
&= (q_0-q_{2k})^2\,\sum\limits_{\ell=1}^{2k-1} q_\ell + \sum\limits_{\ell=1}^{2k-1} q^2_\ell\,(1-q_\ell) + \sum\limits_{\ell=1}^{2k-1} q^2_\ell\,(1-q_\ell-q_0-q_{2k}) \\
&\qquad  + 2\sum\limits_{1\leq i<j<\ell \leq 2k-1} q_i\,q_j\,q_\ell
\end{align*}
and the proof is completed.
\end{proof}

\begin{remark}
In virtue of the content of Proposition \ref{prop:1}, one naturally expects that the solution of \eqref{eqn:ODE_alpha=0} will converge (as $t \to \infty$) to a unique equilibrium (denoted by $\overbar{{\bf q}}$) which belongs to $\mathcal{A}_\gamma$, hence the terminal opinion distribution will be polarized at two extreme opinions $0$ and $2k$. However, the dynamics \eqref{eqn:ODE_alpha=0} does not have any obvious invariants allowing us to link $\overbar{{\bf q}}$ with the initial datum ${\bf q}(0)$. In some sense, the long time behavior of the ODE system \eqref{eqn:ODE_alpha=0} resembles the large time behavior of a self-organized dynamics from mathematical biology \cite{motsch_new_2011}, since the equilibrium distribution $\overbar{{\bf q}}$ is encoded in the underlying system \eqref{eqn:ODE_alpha=0} and depends on the initial condition as well. Whether $\overbar{{\bf q}}$ can be expressed explicitly in terms of ${\bf q}(0)$ is a challenging open problem for future work.
\end{remark}

In order to demonstrate the production of the re-scaled Gini index numerically, we employ again $k=2$ and ${\bf q}(t=0) = (0.25, 0.2, 0.35, 0.2, 0)$, maintaining the same set-up as used in the generation of figure \ref{fig:alpha=1,k=2}. We plot in figure \ref{fig:alpha=0,k=2}-left and figure \ref{fig:alpha=0,k=2}-right the evolution of $G[{\bf q}(t)]$ and the solution vector ${\bf q}(t)$ with respect to time.

\begin{figure}[!htb]
  \begin{subfigure}{0.45\textwidth}
    \centering
    \includegraphics[scale=0.55]{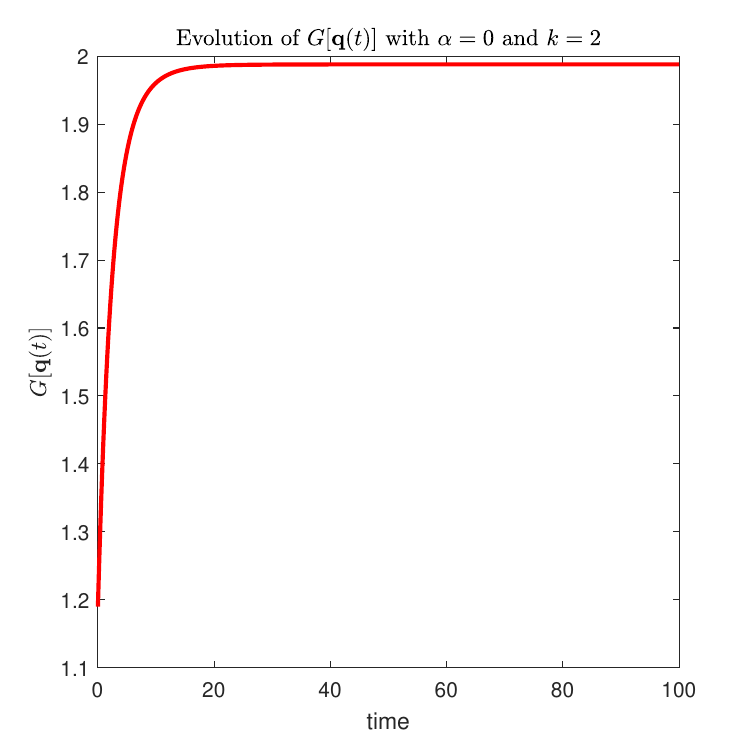}
  \end{subfigure}
  \hspace{0.1in}
  \begin{subfigure}{0.45\textwidth}
    \centering
    \includegraphics[scale=0.55]{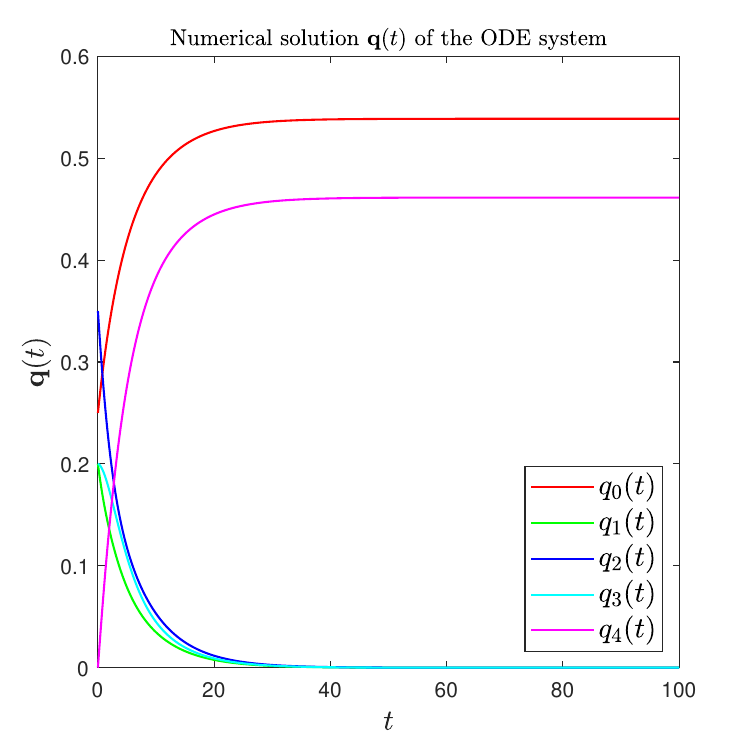}
  \end{subfigure}
  \caption{{\bf Left}: Production of $\tilde{G}[{\bf q}(t)]$ along the solution of \eqref{eqn:ODE_alpha=0}. {\bf Right}: Evolution of the solution vector ${\bf q}(t)$ with respect to time.} 
  \label{fig:alpha=0,k=2}
\end{figure}

\begin{remark}
In the special case where $k=1$, the system \eqref{eqn:ODE_alpha=0} is amenable to explicit solution, leading us to
\begin{equation*}
q_1(t) = \frac{C}{C+\expo^t},\quad q_2(t) = \frac{C+\expo^t}{\expo^t}\left[\frac{C}{2}\,\frac{(C+2)\,\expo^{2t}-2\,\expo^t - C}{(1+C)^2(\expo^t + C)^2}+\frac{q_2(0)}{1+C}\right],
\end{equation*}
and $q_0(t) = 1-q_1(t)-q_2(t)$, in which $C_1 = \frac{q_1(0)}{1-q_1(0)}$. Consequently, we deduce that ${\bf q}(t) \xrightarrow{t\to \infty} \overbar{{\bf q}}$ where $\overbar{{\bf q}}$ is given by
\begin{equation}\label{eq:q_bar}
\overbar{q}_0 = \frac 12\,(1+q^2_0(0)-q^2_2(0)), \quad \overbar{q}_1 = 0, \quad \overbar{q}_2 = \frac 12\,(1-q^2_0(0)+q^2_2(0)).
\end{equation}
It is worth mentioning that even in the simplest case where $k=1$, the equilibrium distribution $\overbar{{\bf q}}$ \eqref{eq:q_bar} already exhibits a nontrivial dependence on the initial datum. Finally, we remark that when $k=1$ the standard Gini index $G[{\bf q}]$ also enjoys a monotonicity property similar to its re-scaled version $\tilde{G}[{\bf q}]$, since we have
\begin{equation*}
\frac{\dd}{\dd t} G[{\bf q}] = \frac{q^3_1\,(1-q_1)+ 4\,q^2_1\,q^2_2 + 2\,q_1\,q^3_2}{(q_1+2\,q_2)^2} \geq 0.
\end{equation*}
However, in general the standard Gini index will no longer serve as a Lyapunov functional for the evolution system \eqref{eqn:ODE_alpha=0} when $k\geq 2$.
\end{remark}

\subsection{Relaxation to uniformly mixed opinions when $\alpha = 1/2$}\label{subsec:3.3}

When $\alpha = \beta = \frac 12$, the nonlinear ODE system \eqref{eqn:ODE_main} becomes
\begin{equation}\label{eqn:ODE_alpha=1/2}
\left\{
\begin{aligned}
q'_0 & = \frac 12\,(1-q_1)\,q_1 - \frac 12\,(1-q_0)\,q_0\\
q'_n & = \frac 12\,q_{n-1}\,(1-q_{n-1}) + \frac 12\,q_{n+1}\,(1-q_{n+1})  - q_n\,(1-q_n), ~~0<n<2k \\
q'_{2k} & = \frac 12\,(1-q_{2k-1})\,q_{2k-1} - \frac 12\,(1-q_{2k})\,q_{2k}
\end{aligned}\right.
\end{equation}
Starting from any ${\bf q}(0) \in \mathcal{S}_\mu$ with $\mu \in (0,2k)$, we easily see that the unique equilibrium solution of \eqref{eqn:ODE_alpha=1/2}, denoted by $\hat{{\bf q}}$, is given by the uniform distribution over $\{0,1,\cdots,2k\}$:
\begin{equation}\label{eq:uniform_equilibrium}
\hat{q}_n = \frac{1}{2k+1},\quad \textrm{for all}~~ 0\leq n \leq 2k.
\end{equation}
In this case, we demonstrate that the relative entropy serves as a Lyapunov functional along the solution of \eqref{eqn:ODE_alpha=1/2}. We recall that for a given ${\bf q} \in \mathcal{P}(\{0,1,\ldots,2k\})$, the relative entropy from ${\bf q}$ to $\hat{{\bf q}}$ is defined by
\begin{equation*}
\DD_{\KL}({\bf q}~||~\hat{{\bf q}}) \coloneqq \sum\limits_{n=0}^{2k} q_n\,\log \frac{q_n}{\hat{q}_n} = \sum\limits_{n=0}^{2k} \hat{q}_n\,\frac{q_n}{\hat{q}_n}\,\log \frac{q_n}{\hat{q}_n}.
\end{equation*}
We aim to show that the relative entropy $\DD_{\KL}({\bf q}~||~\hat{{\bf q}})$ will decay exponentially fast to zero (at least after some finite time) along the solution of system \eqref{eqn:ODE_main}. The key tool towards such exponential decay in relative entropy relies on a logarithmic Sobolev inequality (LSI) for the discrete uniform distribution.

\begin{lemma}\label{lem:LSI}
Assume that $2k \in \mathbb{N}_+$ is given and denote by $\hat{{\bf q}} = (\hat{q}_0,\ldots,\hat{q}_{2k})$ the uniform distribution on $\{0,1,\cdots,2k\}$. Then there exists some universal constant $C = C(k) \propto k^2$ depending only on $k$ such that
\begin{equation}\label{eq:I1}
\sum\limits_{n=0}^{2k} \hat{q}_n\,f^2_n\,\log f^2_n \leq C\,\sum\limits_{n=0}^{2k-1} \hat{q}_n\,(f_{n+1}-f_n)^2
\end{equation}
for all ${\bf f} = (f_0,\ldots,f_{2k}) \in \mathbb{R}^{2k+1}_+$ satisfying $\sum\limits_{n=0}^{2k} \hat{q}_n\,f^2_n = 1$.
\end{lemma}

The proof of this classical result can be found in \cite{diaconis_logarithmic_1996,matthes_structure_2023}, and we remark here that the LSI \eqref{eq:I1} is merely a discrete analog of the LSI for the uniform measure on a one-dimensional compact interval, which takes the following form \cite{ghang_sharp_2014}:
\begin{equation}\label{eq:I2}
\int_0^{2k} f^2(x)\,\log f^2(x)\,\mu(x)\,\dd x \leq C\,\int_0^{2k} |f'(x)|^2\,\mu(x)\,\dd x,
\end{equation}
where $C = C(k) \propto k^2$, $\mu(x)$ is the uniform distribution on $[0,2k]$, and $f \colon [0,2k] \to \mathbb{R}_+$ is any smooth function satisfying the constraint that $\int_0^{2k} f^2(x)\,\mu(x)\,\dd x= 1$.

\begin{theorem}[Entropy dissipation]\label{thm:3}
For any $k \in \mathbb{N}_+$, if ${\bf q}(t)$ is a solution of the nonlinear system of ODEs \eqref{eqn:ODE_alpha=1/2} with ${\bf q}(0) \in \mathcal{S}_\mu$ and $\mu \in (0,2k)$, then there exist some $\delta \in (0,1)$ and some finite time $t_* > 0$ for which
\begin{equation}\label{eq:bound_in_relative_entropy}
\DD_{\KL}({\bf q}(t)~||~\hat{{\bf q}})\leq \DD_{\KL}({\bf q}(t_*)~||~\hat{{\bf q}})\,\expo^{-\frac{2\,\delta}{C}\,(t-t_*)},\quad \forall~t \geq t_*,
\end{equation}
where $C= C(k) \propto k^2$.
\end{theorem}

\begin{proof}
We notice that the relative entropy is dissipating along the solution of \eqref{eqn:ODE_alpha=1/2} since
\begin{equation*}
\begin{aligned}
&\frac{\dd}{\dd t} \DD_{\KL}({\bf q}~||~\hat{{\bf q}}) = \sum\limits_{n=0}^{2k} q'_n\,\log q_n = q'_0\,\log q_0 + q'_{2k}\,\log q_{2k} + \sum\limits_{n=1}^{2k-1} q'_n\,\log q_n \\
&= \frac 12\left[(1-q_1)\,q_1 - (1-q_0)\,q_0\right]\,\log q_0 + \frac 12\left[(1-q_{2k-1})\,q_{2k-1} - (1-q_{2k})\,q_{2k}\right]\,\log q_{2k} \\
&\quad + \frac 12\,\left[\sum\limits_{n=1}^{2k-1} \left(q_{n+1}\,(1-q_{n+1})-q_n\,(1-q_n)\right)\,\log q_n - \sum\limits_{n=1}^{2k-1} \left(q_n\,(1-q_n)-q_{n-1}\,(1-q_{n-1})\right)\,\log q_n \right] \\
&= \frac 12\,\sum\limits_{n=0}^{2k-1} \left(q_{n+1}\,(1-q_{n+1})-q_n\,(1-q_n)\right)\,\log q_n - \frac 12\,\sum\limits_{n=1}^{2k} \left(q_n\,(1-q_n)-q_{n-1}\,(1-q_{n-1})\right)\,\log q_n \\
&= \frac 12\,\sum\limits_{n=0}^{2k-1} \left(q_{n+1}\,(1-q_{n+1})-q_n\,(1-q_n)\right)\,\log \frac{q_n}{q_{n+1}} \\
&= \frac 12\,\sum\limits_{n=0}^{2k-1} (1-q_n-q_{n+1})\,(q_{n+1}-q_n)\,\left(\log q_{n+1} - \log q_n\right) \leq 0.
\end{aligned}
\end{equation*}
Consequence, the relative entropy is a Lyapunov functional for the finite-dimensional system of nonlinear ODEs \eqref{eqn:ODE_alpha=1/2}, and we obtain the pointwise convergence guarantee ${\bf q}(t) \xrightarrow{t\to \infty} \hat{{\bf q}}$ (using a similar argument as in \cite{cao_derivation_2021}). In particular, the aforementioned qualitative convergence ensures the existence of some $\delta \in (0,1)$ and some finite $t_* > 0$ such that
\begin{equation*}
\max\limits_{0\leq n\leq 2k-1} \left(q_n(t) + q_{n+1}(t)\right) \leq 1-\delta \quad \forall~ t\geq t_*
\end{equation*}
or equivalently that
\begin{equation}\label{eq:ingredients}
\min\limits_{0\leq n\leq 2k-1} \left(1 - q_n(t) - q_{n+1}(t)\right) \geq \delta \quad \forall~ t\geq t_*.
\end{equation}
Now we observe for all $a,b\in \mathbb{R}_+$ that
\[(a-b)\,(\log a - \log b) = \int_b^a 1\,\dd t\cdot \int_b^a \frac{1}{t}\,\dd t \geq \left(\int_b^a \frac{1}{\sqrt{t}}\,\dd t\right)^2 = 4(\sqrt{a}-\sqrt{b})^2.\]
Therefore we can invoke Lemma \ref{lem:LSI} with $f_n = \sqrt{\frac{q_n}{\hat{q}_n}}$ to deduce that
\begin{equation}\label{eq:ready}
\frac{\dd}{\dd t} \DD_{\KL}({\bf q}~||~\hat{{\bf q}}) \leq -2\,\delta\,\sum\limits_{n=0}^{2k-1} \hat{q}_n\,\left(\sqrt{\frac{q_{n+1}}{\hat{q}_{n+1}}}-\sqrt{\frac{q_n}{\hat{q}_n}}\right)^2 \leq -\frac{2\,\delta}{C}\,\DD_{\KL}({\bf q}~||~\hat{{\bf q}})
\end{equation}
for all $t\geq t_*$, where $C = C(k) \propto k^2$. Thanks to Gronwall's inequality, we reach the advertised bound \eqref{eq:bound_in_relative_entropy}.
\end{proof}

To illustrate the quantitative convergence result proved in Theorem \ref{thm:3} (with $k=2$), we use ${\bf q}(t=0) = (0.25, 0.2, 0.35, 0.2, 0)$ as the initial datum, with the same set-up as used in the generation of figures \ref{fig:alpha=1,k=2} and \ref{fig:alpha=0,k=2}. We plot in figure \ref{fig:alpha=.5,k=2,ex1}-left and figure \ref{fig:alpha=.5,k=2,ex1}-right the evolution of $\DD_{\KL}({\bf q}(t)~||~\hat{{\bf q}})$ in the normal scale and the semi-logy scale, starting from ${\bf q}(t=0) = (0.25, 0.2, 0.35, 0.2, 0)$.

\begin{figure}[!htb]
  \begin{subfigure}{0.45\textwidth}
    \centering
    \includegraphics[scale=0.4]{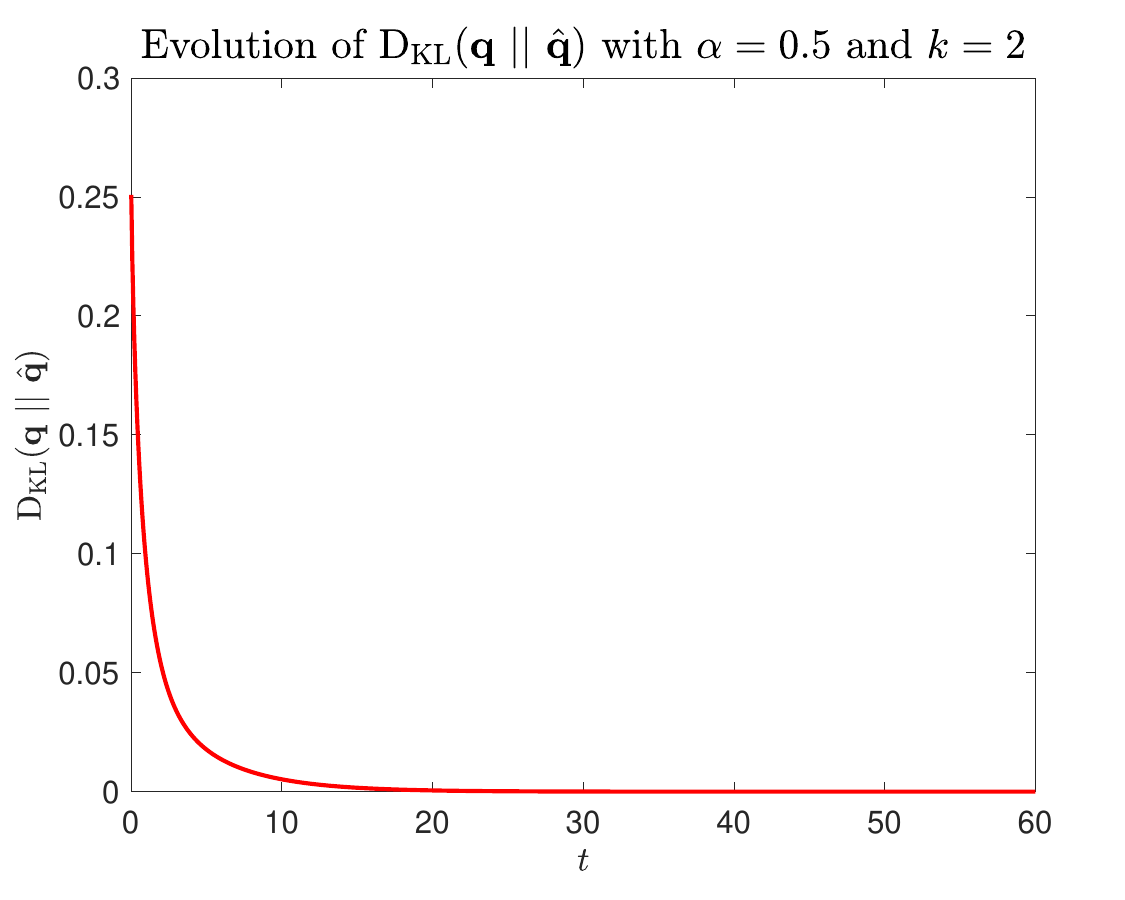}
  \end{subfigure}
  \hspace{0.1in}
  \begin{subfigure}{0.45\textwidth}
    \centering
    \includegraphics[scale=0.4]{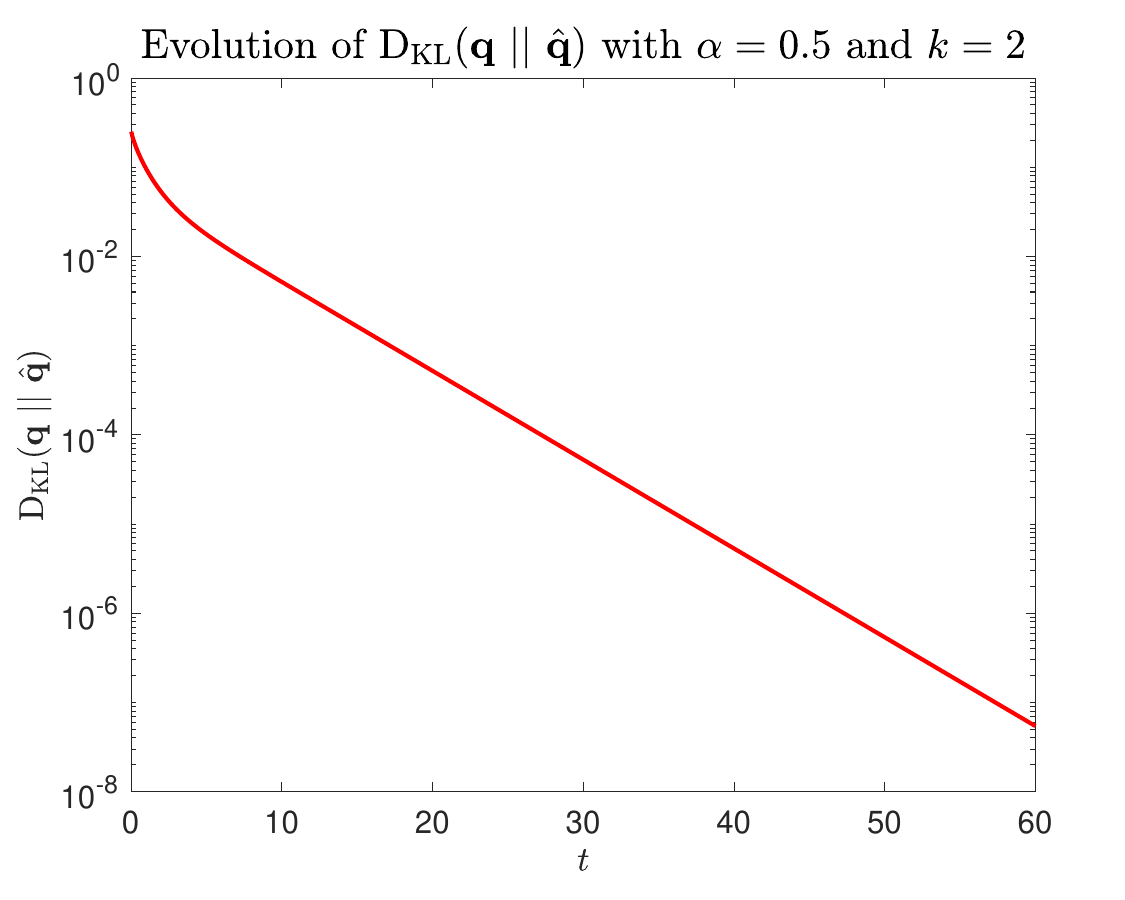}
  \end{subfigure}
   \caption{{\bf Left}: Decay of $\DD_{\KL}({\bf q}(t)~||~\hat{{\bf q}})$ along the solution of \eqref{eqn:ODE_alpha=1/2} with $k=2$. {\bf Right}: Decay of $\DD_{\KL}({\bf q}(t)~||~\hat{{\bf q}})$ in the semi-logy scale. The decay is exponentially fast with respect to time, as predicted by Theorem \ref{thm:3}.}
  \label{fig:alpha=.5,k=2,ex1}

\end{figure}

Next, we show that an estimate similar to \eqref{eq:bound_in_relative_entropy} for the relative entropy can also be established for another ``pseudo-metric'' known as the chi-squared distance, defined via
\begin{equation*}
\chi^2({\bf f},{\bf g}) = \sum\limits_{n=0}^{2k} \frac{(f_n-g_n)^2}{g_n}
\end{equation*}
whenever ${\bf f},{\bf g} \in \mathcal{P}(\{0,1,\ldots,2k\})$ such that $g_n > 0$ for all $0\leq n\leq 2k$.

\begin{theorem}\label{thm:4}
Under the settings and notations of Theorem \ref{thm:3}, for some $\tilde{C} = \tilde{C}(k)$ we have
\begin{equation}\label{eq:bound_in_relative_entropy}
\chi^2({\bf q}(t),\hat{{\bf q}}) \leq \chi^2({\bf q}(t_*),\hat{{\bf q}})\,\expo^{-\frac{\delta}{\tilde{C}}\,(t-t_*)},\quad \forall~t \geq t_*.
\end{equation}
\end{theorem}

\begin{proof}
We observe that the chi-squared distance $\chi^2({\bf q},\hat{{\bf q}})$ also serves as a Lyapunov functional for the ODE system \eqref{eqn:ODE_alpha=1/2} due to the following computations:
\begin{align*}
&\frac{\dd}{\dd t} \chi^2({\bf q},\hat{{\bf q}}) = 2\,(2k+1)\,\sum\limits_{n=0}^{2k} (q_n-\hat{q}_n)\,q'_n \\
&= 2\,(2k+1)\,\left[q_0\,q'_0 + q_{2k}\,q'_{2k} + \sum\limits_{n=1}^{2k-1} q_n\,q'_n\right] \\
&= -(2k+1)\,\sum\limits_{n=0}^{2k-1} (q_{n+1}-q_n)\,\left[q_{n+1}\,(1-q_{n+1})-q_n\,(1-q_n)\right] \\
&= -(2k+1)\,\sum\limits_{n=0}^{2k-1} (1-q_n-q_{n+1})\,(q_{n+1}-q_n)^2 \\
&= -\sum\limits_{n=0}^{2k-1} (1-q_n-q_{n+1})\,\hat{q}_n\,\left(\frac{q_{n+1}}{\hat{q}_{n+1}}-\frac{q_n}{\hat{q}_n}\right)^2 \leq 0.
\end{align*}
Thanks to the previous estimate \eqref{eq:ingredients} and the Poincar{\'e} inequality satisfied by the uniform distribution $\hat{\bf {q}}$ \cite{bobkov_discrete_1999}, we deduce that
\begin{equation}\label{eq:ready_to_go}
\frac{\dd}{\dd t} \chi^2({\bf q},\hat{{\bf q}}) \leq -\delta\,\sum\limits_{n=0}^{2k-1}\hat{q}_n\,\left(\frac{q_{n+1}}{\hat{q}_{n+1}}-\frac{q_n}{\hat{q}_n}\right)^2 \leq -\frac{\delta}{\tilde{C}}\,\chi^2({\bf q},\hat{{\bf q}}) \quad \forall~t \geq t_*
\end{equation}
for some $\tilde{C} = \tilde{C}(k) > 0$. Finally, a routine application of Grownall's lemma leads us to the claimed bound.
\end{proof}

To demonstrate the quantitative bound established in Theorem \ref{thm:4} (with $k=2$), we use the same set-up as used in the generation of figures \ref{fig:alpha=1,k=2}, \ref{fig:alpha=0,k=2} and \ref{fig:alpha=.5,k=2,ex1}. We plot in figure \ref{fig:alpha=.5,k=2,ex2}-left and figure \ref{fig:alpha=.5,k=2,ex2}-right the evolution of $\chi^2({\bf q}(t),\hat{{\bf q}})$ in the normal scale and the semi-logy scale, starting from ${\bf q}(t=0) = (0.25,0.2,0.35,0.2,0)$.

\begin{figure}[!htb]
  \begin{subfigure}{0.45\textwidth}
    \centering
    \includegraphics[scale=0.4]{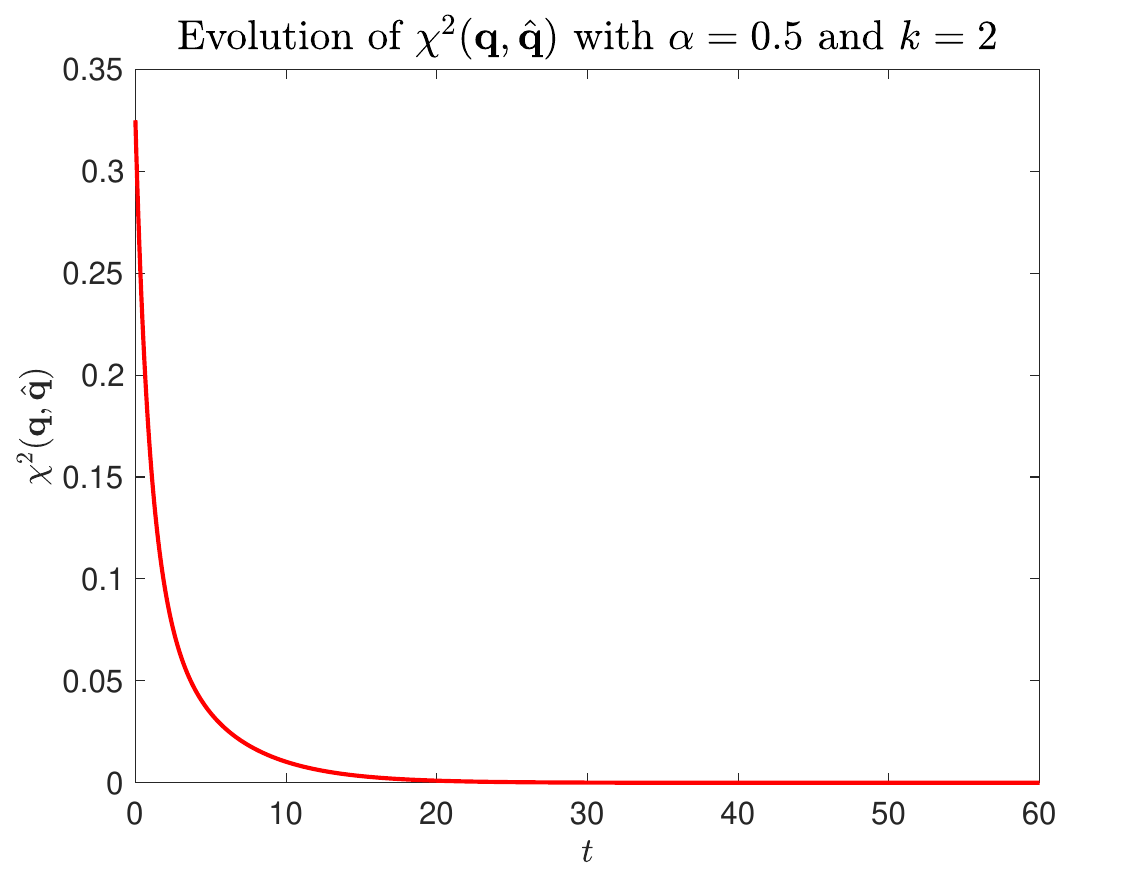}
  \end{subfigure}
  \hspace{0.1in}
  \begin{subfigure}{0.45\textwidth}
    \centering
    \includegraphics[scale=0.4]{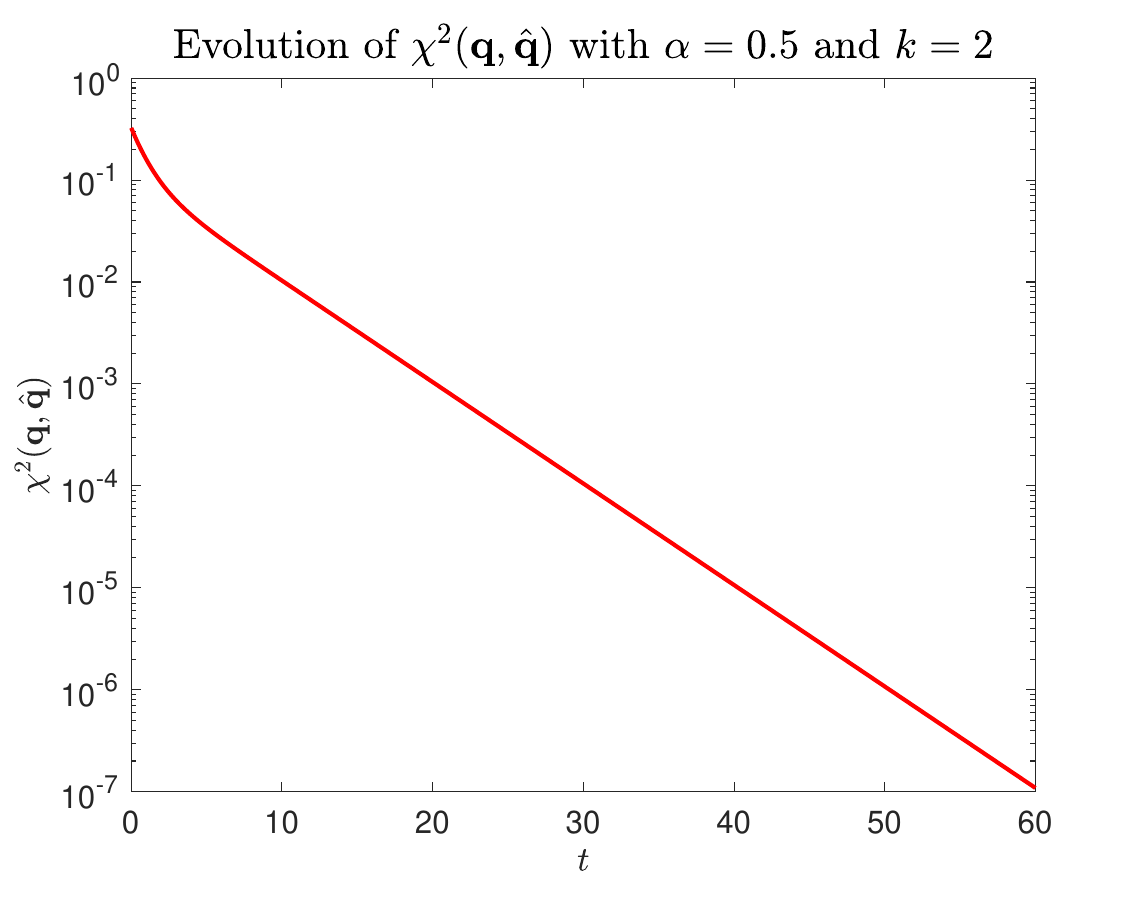}
  \end{subfigure}
   \caption{{\bf Left}: Decay of $\chi^2({\bf q}(t),\hat{{\bf q}})$ along the solution of \eqref{eqn:ODE_alpha=1/2} with $k=2$. {\bf Right}: Decay of $\chi^2({\bf q}(t),\hat{{\bf q}})$ in the semi-logy scale. The decay is exponentially fast with respect to time, as shown by Theorem \ref{thm:4}.}
  \label{fig:alpha=.5,k=2,ex2}
\end{figure}


\section{Conclusion}
\setcounter{equation}{0}

In this work, we proposed and analyzed the Iterative Persuasion-Polarization (IPP) opinion model in the mean field region as the number of agents tends to infinity. Our model contributes to the growling list of opinion dynamics among the sociophysics literature and contains a parameter $\alpha \in [0,1]$ measuring the tendency that each agent will align his/her opinion with another agent's opinion during an interaction process. We provided analytical and quantitative results regarding the large time behavior of the mean-field IPP ODE system \eqref{eqn:ODE_main} under three particular choices of the parameter $\alpha$. In particular, we proved that the steady state opinion distribution is a two-point distribution supported near the average initial opinion when $\alpha = 1$, indicating the formation of a ``almost consensus'' opinion profile. On the other hand, we showed when $\alpha = 0$ that the opinion distribution converges to a polarized state in which only two extreme opinions survive in the long run. Lastly, in the case where $\alpha = 1/2$, we established the convergence to a uniform distribution for solutions of the mean-field system of ODEs \eqref{eqn:ODE_main} under the large time limit. The present paper also leaves many important unsolved problems suitable for further research activities in the future. First, is it possible analyze the large time behaviour of the nonlinear ODE system \eqref{eqn:ODE_main} when $\alpha \in [0,1] \setminus \{0,1/2,1\}$~? If so, can we determine the equilibrium distribution of opinions~? Numerical solutions of the ODE system in this case suggest that the system will converge to a unique equilibrium regardless of initial datum that depends only on $\alpha$ and $k$ (as it does with $\alpha = 1/2$) as illustrated in figure \ref{fig:opinionDist}.

\begin{figure}[!htb]
\centering
\includegraphics[width=0.8\textwidth]{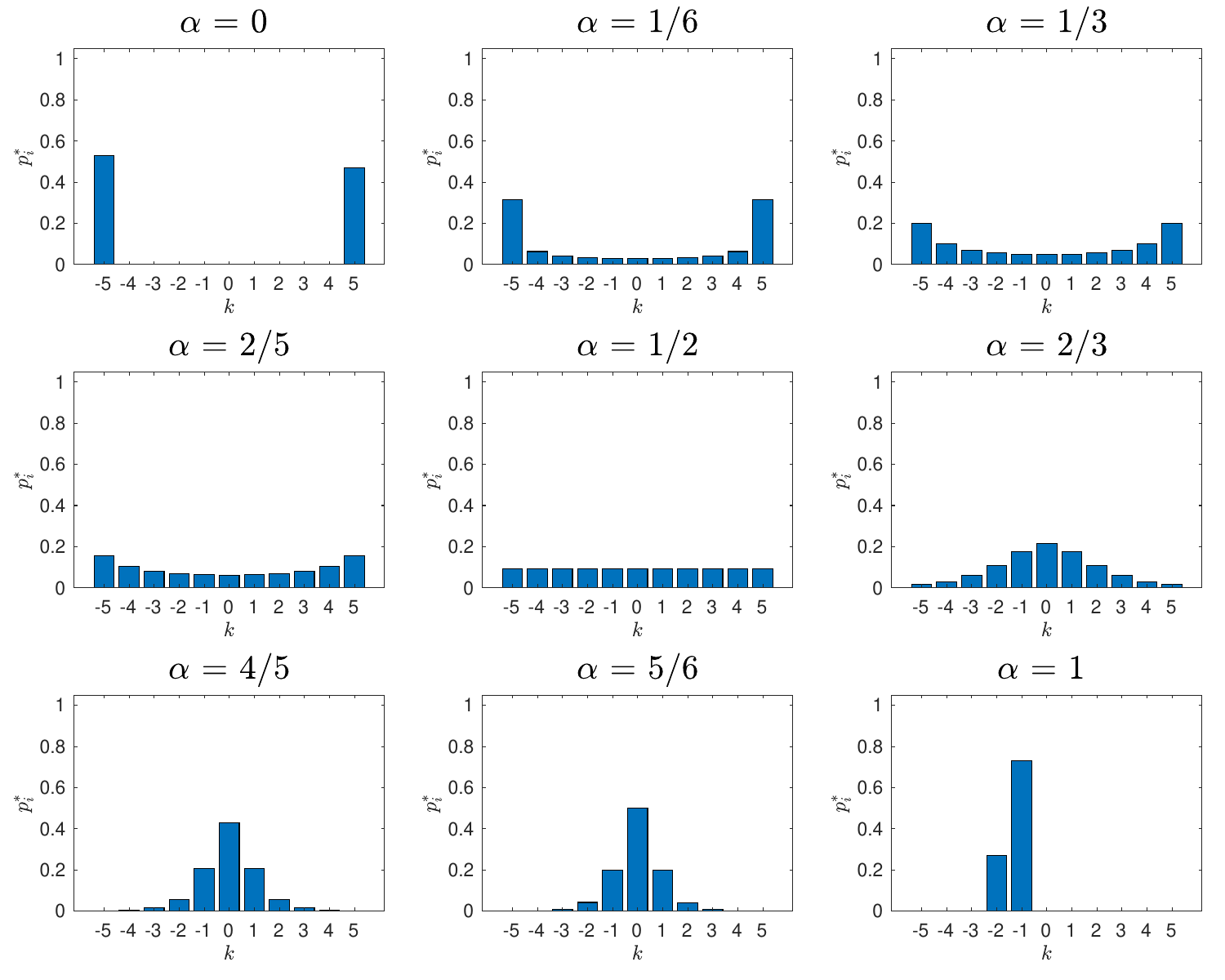}
\caption{Distribution of opinions at equilibrium (with $k=5$) for varying values of $\alpha \in [0,1]$, starting from the initial datum ${\bf q}(t=0)$ given by ${\bf q}(t=0)=(0.25,0.15,0.05,0.05,0.15,0,0.10,0.05,0.05,0.05,0.10)$.}
\label{fig:opinionDist}
\end{figure}

Second, in the case of $\alpha = 0$, how can we link the equilibrium polarized opinion profile with the initial opinion distribution so that a more explicit form of the equilibrium distribution can be identified~? A proper theoretical treatment of these questions allows us have a better understanding about the roles played by the persuasion parameter $\alpha$ and (possibly) the initial datum on the shape of the steady state distribution of opinions.


\end{document}